\documentclass[10pt]{article}
\usepackage{}
\usepackage{amsfonts}
\usepackage{amssymb,amsmath,amsthm, amsfonts}

\usepackage[numbers,sort&compress]{natbib}

\providecommand{\U}[1]{\protect \rule{.1in}{.1in}}
\newtheorem{theorem}{Theorem}

\newtheorem{definition}[theorem]{Definition}
\newtheorem{example}[theorem]{Example}

\newtheorem{lemma}[theorem]{Lemma}

\newtheorem{remark}[theorem]{Remark}

\begin{document}

 \title { An optimal control problem of forward-backward stochastic Volterra integral equations with state constraints}

\author{Qingmeng Wei\\
{\small
School of Mathematics and Statistics, Northeast Normal University, Changchun 130024,  China}\\
{\small {\it Email: qingmengwei@gmail.com}}\\
Xinling Xiao\\
{\small School of Mathematics, Shandong University, Jinan 250100, China.}\\
{\small{\it Email:xinlingxiao@mail.sdu.edu.cn}}}

\date{}

\maketitle

\noindent \textbf{Abstract.} This paper is devoted to the stochastic
optimal control problems for systems governed by forward-backward
stochastic Volterra integral equations (FBSVIEs, for short) with
state constraints. Using Ekeland's variational principle, we obtain
one kind of variational inequality. Then, by dual method, we derive
a stochastic maximum principle which gives the necessary conditions
for the optimal controls.
 \bigskip


\section{Introductin}

As we known, with the exception of the applications in biology,
physical, etc, Volterra integral equations often appear in some
mathematical economic problems, for example, the relationships
between capital and investment which include memory effects (in
\cite{KM}, the present stock of capital depends on the history of
investment strategies over a period of time). And the simplest way
to describe such memory effects is through Volterra integral
operators. Based on the importance of Volterra integral equations,
we will study an stochastic optimal control problem about a class of
nonlinear stochastic equations$-$forward-backward stochastic
Volterra integral equations (FBSVIEs, for short). First we review
the backgrounds of these two kinds of Volterra integral equations:
forward stochastic Volterra integral equations (FSVIEs, for short)
and backward stochastic Volterra integral equations (BSVIEs, for
short).

Let $B(\cdot )$ be a standard $d$-dimensional Brownian motion defined on a
complete filtered probability space $(\Omega ,\mathcal{F},\mathbb{F},P)$,
where $\mathbb{F}=\{\mathcal{F}_{t}\}_{t\geq 0}$ is its natural filtration
generated by $B(\cdot )$ and augmented by all the $P$-null sets in $\mathcal{%
F}$. Consider the following FSVIE:
\begin{equation*}
X(t)=\varphi (t)+\int_{0}^{t}b(t,s,X(s))ds+\int_{0}^{t}\sigma
(t,s,X(s))dB_{s},\ t\in \lbrack 0,T].
\end{equation*}%
The readers may refer to \cite{BM, D, K, EPPP, PP} and the reference
cited therein, for the general results on FSVIEs. When studying the
stochastic optimal control problems for FSVIEs, we need one kind of
adjoint equation in order to derive a stochastic maximum principle.
This new adjoint equation is actually a linear BSVIE. This motivates
the investigation of the theory and applications of BSVIEs.

The following BSVIE was firstly introduced by Yong \cite{JY2}:
\begin{equation}
Y(t)=\psi(t)+{ \int_{t}^{T}}g(t,s,Y(s),Z(s,t))ds-{
\int_{t}^{T}}Z(t,s)dB_{s},\quad t\in \lbrack0,T],  \label{equ1}
\end{equation}
where $g:\Delta^{c}\times\mathbb{R}^{m}\times\mathbb{R}^{m\times
d}\times\Omega\rightarrow\mathbb{R}^{m}$ and $\psi:[0,T]\times\Omega
\rightarrow\mathbb{R}^{m}$ are given maps with $\Delta^{c}=\{(t,s)\in
\lbrack0,T]^{2}|t<s\}$. For each $t\in\lbrack0,T]$, $\psi(t)$ is $\mathcal{F}%
_{T}$-measurable (Lin \cite{JL} studied (\ref{equ1}) when $%
\psi(\cdot)\equiv\xi$). It is obvious that BSVIE is a natural
generalization of backward stochastic differential equation (BSDE,
for short). Comparing with BSDEs, BSVIE still has its own features
as listed in Yong \cite{JY2, JY}. One of the advantages is to study
time-inconsistent phenomenon. As shown in Laibson~\cite{DL} and
Strotz~\cite{RS}, in the real world, time-inconsistent preference
usually exists. At this point, one needs BSVIEs to generalize the
so-called stochastic differential utility in~\cite{DDLE} and dynamic
risk measures (see~\cite{PAFD, PCDFKM, SP3, TW}). Other applications
are in the non-exponential discounting problems (see Ekeland, Lazrak
\cite{IEAL}, Ekeland, Pirvu \cite{IETP}) and time-inconsistent
optimal control problem (see Yong \cite{JY3, JY20121}). In
\cite{JY3, JY20121}, Yong solved a time-inconsistent optimal control
problem by introducing a family of $N$-person non-cooperative
differential games, and got an equilibrium control which was
represented via a forward ordinary differential equation with a
backward Riccati-Volterra integral equation.

As stated in Yong \cite{JY4}, $\psi(t)$ in BSVIE (\ref{equ1}) could
represent the total (nominal) wealth of certain portfolio which
might be a combination of certain  contingent claims (for example,
European style, which is mature at time $T$, are usually only
$\mathcal {F}_T$-measurable), some current cash flows, positions of
stocks, mutual funds, and bonds, and so on, at time $t$. So, in
general, the position process $\psi(\cdot)$ is not necessarily
$\mathbb{F}$-adapted, but a stochastic process merely $\mathcal
{F}_T$-measurable. And Yong gave an example to make this point more
clear in \cite{JY4}. Focusing on this kind of position process
$\psi(\cdot)$, a class of convex/coherent dynamic risk measures was
introduced by Yong in \cite{JY4} to measure the risk dynamically.
Hence, one kind of control problem appears:  how to minimize the
risk, or how to maximize the utility. Wang, Shi \cite{WS} obtained a
maximum principle for FBSVIEs without state constraints. In this
paper, we study one kind of optimal control problem in which the
state equations are governed by the following FBSVIEs:
\begin{equation}
\left\{
\begin{array}{llll}
X(t)  =  f(t)+\int_{0}^{t}b(t,s,X(s),u(s))ds+{%
 \int_{0}^{t}}\sigma (t,s,X(s),u(s))dB_{s},\quad t\in \lbrack
0,T], &  \\
Y(t) =  \psi (t)+{ \int_{t}^{T}}%
g(t,s,X(s),Y(s),Z(s,t),u(s))ds-{ \int_{t}^{T}}Z(t,s)dB_{s}. &
\end{array}%
\right.  \label{equ1.1}
\end{equation}%
By choosing admissible controls $(u,\psi )$, we shall maximize the
following
objective functional%
\begin{equation}\label{cost}
\begin{array}{llll}
J(\psi ,u) :=  E [\int_{0}^{T}{ \int_{t}^{T}%
}l_{2}(t,s,X(s),Y(s),Z(s,t),u(s))dsdt &  \\
\qquad\qquad\quad+\int_{0}^{T}{ \int_{0}^{t}%
}l_{1}(t,s,X(s),u(s))dsdt+\int_{0}^{T}q(\psi
(t))dt+h(X(T))+\int_0^T k(Y(s))ds ]. &
\end{array}%
\end{equation}

Our formulation has the following new features:

%(i) We introduce a new type of controls $u(\cdot)$ which
%depend on both time variables $t$ and $s$. In classical optimal
%control problems,
%controls are usually assumed to depend on only one time variable (see \cite%
%{BE, C, WS, JY2, JY4, JY, JY3} and the references therein). In our
%model, as time goes by, the controlled system is changing, therefore
%it is reasonable to keep changing the controls to get more better
%objectives. The key observation is that for systems described by
%SVIEs, using $u(\cdot)$ is different from using classical
%controls $u(\cdot )$. In Section 4, we give two examples to show the
%difference. Now we haven't found the finance applications for such
%kind of controls. Some further studies on this point are still under
%consideration.

(i) A strong assumption that $g(t,\cdot ,\cdot ,\cdot ,\cdot ,\cdot )$ in (%
\ref{equ1.1}) is $\mathcal{F}_{t}$-measurable is given in \cite{WS}. By
applying the duality principle introduced in Yong \cite{JY4}, we overcome
this restriction and assume a natural condition that $g(\cdot ,s,\cdot
,\cdot ,\cdot ,\cdot )$ is $\mathcal{F}_{s}$-measurable.

(ii) $\psi $ in (\ref{equ1.1}) is the terminal state of the BSVIE.
In our
formulation $\psi $ is also regarded as a control and our control is a pair $%
(u,\psi )$. In mathematical finance, such kind of controls often
appears as ``consumption-investment plan" (see \cite{pliska}). For
the recent progress
of studying this kind of control we refer the reader to \cite{NESPMC, SJSP, SJXZ, SJ2}. We also impose constraints on the state process $Y(\cdot )$ and $\psi $.

(iii) We consider the double integral in the cost functional (\ref{cost}) in theory. Some further studies on the applications are still under consideration. 

In order to solve this optimal control problem, we adopt the
terminal perturbation method, which was introduced in \cite{NESPMC, JI-2, JI,
SJSP, JIWEI, JIWEI-1, JIWEIZ, JW, JZ-0, SJXZ}. 
Recently, the dual
approach is applied to utility optimization problem with volatility
ambiguity (see \cite{Eji,Eji-1}). 
The basic idea is to perturb the terminal state $\psi $
and $u$ directly. By applying Ekeland's variational principle to
tackle the state constraints, we derive a stochastic maximum
principle which characterizes the optimal control. It is worth to
point out that in place of It\^{o}'s formula, we need two
duality principles established by Yong in \cite{JY4, JY} to obtain
the above results.

This paper is organized as follows. First, we recall some elements
of the theory of BSVIEs in Section 2. In Section 3, we formulate the
stochastic optimization problem and prove a stochastic maximum
principle. In Section 4, we give two examples. The first example is  associated with the model we studied. The last
example is about the `terminal' control $\psi (\cdot )$,

\section{Preliminaries}

Let $B(\cdot)$ be a $d$-dimensional Brownian motion defined on a complete
filtered probability space $(\Omega ,\mathcal{F},\mathbb{F},P)$, where $%
\mathbb{F}=\{\mathcal{F}_{t}\}_{t\ge0}$ is natural filtration generated by $%
B(\cdot)$ and augmented by all the $P$-null sets in $\mathcal{F}$, i.e.,
\begin{equation*}
\mathcal{F}_t=\sigma\{B_r,r\leq t\}\vee \mathcal{N}_P,\ t\in [0,T],
\end{equation*}
where $\mathcal{N}_P$ is the set of all $P$-null sets.

\subsection{Notations}

Here we keep on the definitions and notations for the spaces
introduced in Yong \cite{JY}.

For any $0\leq R< S\leq T$, we denote

$\left\{
\begin{array}{l}
\Delta[R,S]= \{(t,s)\in[R,S]^{2} |R\leq s\leq t\leq S \}, \\
\Delta^{c}[R,S]= \{(t,s)\in[R,S]^{2} |R\leq t< s\leq S \}\equiv
[R,S]^{2}\setminus\Delta[R,S].%
\end{array}
\right. $

For any $A,\ B\in \mathbb{R}^{m\times d}$, define the inner product $\langle
A,B\rangle :=$tr$[AB^{T}]$ and
\begin{equation*}
|A|^{2}=\sum_{j=1}^{d}|a_{j}|^{2}=\sum_{i=1}^{m}\sum_{j=1}^{d}a_{ij}^{2},%
\quad \forall A\equiv (a_{1},\cdots ,a_{d})\equiv (a_{ij})\in \mathbb{R}%
^{m\times d}.
\end{equation*}%
Let $S\in \lbrack 0,T]$, define the following spaces:\newline

\begin{itemize}
\item $L_{\mathcal{F}_{S}}^{p}(0,T):= \{\varphi :[0,T]\times \Omega
\rightarrow \mathbb{R}^{m} |\varphi (\cdot )\text{ is }\mathcal{B}%
([0,T])\otimes \mathcal{F}_{S}\text{-measurable and }
E\int_{0}^{T}|\varphi (t)|^{p}dt<\infty  \};$

\item $L_{\mathbb{F}}^{p}(0,T):= \{\varphi :[0,T]\times \Omega
\rightarrow \mathbb{R}^{m} |\varphi (\cdot )\text{ is }\mathbb{F}\text{%
-adapted and }E\int_{0}^{T}|\varphi (t)|^{p}dt<\infty  \};$

\item $L^{p}(0,T;L_{\mathbb{F}}^{2}(0,T)):= \{Z:[0,T]^{2}\times \Omega
\rightarrow \mathbb{R}^{m\times d} |\text{for almost all }t\in
\lbrack 0,T],Z(t,\cdot )\in L_{\mathbb{F}}^{2}(0,T),$

\begin{equation*} \int_{0}^{T}E ({ \int_{0}^{T}%
}|Z(t,s)|^{2}ds )^{\frac{p}{2}}dt<\infty  \}; \end{equation*}

\item $L_{\mathbb{F}}^{\infty }(0,T;\mathbb{R}^{n}):= \{\varphi
:[0,T]\times \Omega \rightarrow \mathbb{R}^{n} |\mathop{\rm esssup}%
\limits_{\omega \in \Omega }\mathop{\rm sup}\limits_{s\in \lbrack
0,T]}\varphi (s,\omega )<\infty  \};$

\item $L^{\infty }([0,T];L_{\mathbb{F}}^{\infty }(0,T;\mathbb{R}^{n\times
n}))$\newline
$:= \{Z(t,\cdot )\in L_{\mathbb{F}}^{\infty }(0,T;\mathbb{R}^{n\times n})%
 |\mathop{\rm esssup}\limits_{\omega \in \Omega }\mathop{\rm sup}%
\limits_{t\in \lbrack 0,T]}\mathop{\rm sup}\limits_{s\in \lbrack
0,T]}Z(t,s,\omega )<\infty  \};$

\item $\mathcal{H}^{p}[S,T]:=L_{\mathbb{F}}^{p}(S,T)\times L^{p}(S,T;L_{%
\mathbb{F}}^{2}(S,T)).$
\end{itemize}

\subsection{Backward Stochastic Volterra Integral Equations}

For the reader's convenience, we present some results of BSVIEs
which we will use later.

Consider the following integral equation
\begin{equation}
Y(t)=\psi (t)+{ \int_{t}^{T}}g(t,s,Y(s),Z(s,t))ds-{
\int_{t}^{T}}Z(t,s)dB_{s},\quad t\in \lbrack 0,T],  \label{equ2.1}
\end{equation}%
where $\psi (\cdot )\in L_{\mathcal{F}_{T}}^{2}(0,T).$

We assume:

\begin{description}
\item[$\left( \mathbf{H}\right) $] Let $g:\Omega \times \Delta
^{c}[0,T]\times \mathbb{R}^{m}\times \mathbb{R}^{m\times
d}\rightarrow \mathbb{R}^{m}$ be $\mathcal{F}_{T}\otimes
\mathcal{B}(\Delta ^{c}\times \mathbb{R}^{m}\times
\mathbb{R}^{m\times d})$-measurable such that $s\mapsto
g(t,s,y,\zeta )$ is $\mathbb{F}$-progressively measurable for all $
(t,y,\zeta )\in \lbrack 0,T]\times \mathbb{R}^{m}\times
\mathbb{R}^{m\times
d}$ and $$E\int_{0}^{T} ({ \int_{t}^{T}}%
|g(t,s,0,0)|ds )^{2}dt<\infty .$$ Moreover, $\forall (t,s)\in
\Delta ^{c}[0,T],\ (y,\zeta )$ and$\ (\bar{y},\bar{\zeta})\in
\mathbb{R}^{m}\times \mathbb{R}^{m\times d}$,
\begin{equation*}
|g(t,s,y,\zeta )-g(t,s,\bar{y},\bar{\zeta})|\leq L(t,s)(|y-\bar{y}|+|\zeta -%
\bar{\zeta}|),\ \text{ a.s.,}
\end{equation*}%
where $L:\Delta ^{c}[0,T]\rightarrow \mathbb{R}$ is a deterministic function
such that
\begin{equation*}
\mathop{\rm sup}\limits_{t\in \lbrack 0,T]}{ \int_{t}^{T}}%
L(t,s)^{2+\varepsilon }ds<\infty ,\text{ for some }\varepsilon >0.
\end{equation*}
\end{description}

The following $M$-solution of BSVIEs was introduced by Yong \cite{JY} .

\begin{definition}
Let $S\in \lbrack 0,T)$. A pair $(Y(\cdot ),Z(\cdot ,\cdot ))\in \mathcal{H}%
^{2}[S,T]$ is called an adapted $M$-solution of BSVIE (\ref{equ2.1}) on $%
[S,T]$ if (\ref{equ2.1}) holds in the usual It\^{o} sense for almost all $%
t\in \lbrack S,T]$ and, in addition, the following equation holds:
\begin{equation*}
Y(t)=E[Y(t)|\mathcal{F_{S}}]+{ \int_{S}^{t}}Z(t,s)dB_{s},\ \
a.e.,\ t\in \lbrack S,T].
\end{equation*}
\end{definition}

For the proof of the following wellposedness results, the readers
are referred to  Yong \cite{JY}.

\begin{lemma}
\label{th1} Let $\left( \mathbf{H}\right) $ holds. Then for any $\psi (\cdot
)\in L_{\mathcal{F}_{T}}^{2}(0,T)$, BSVIE (\ref{equ2.1}) admits a unique
adapted $M$-solution $(Y(\cdot ),Z(\cdot ,\cdot ))\in \mathcal{H}^{2}[0,T]$
on $[0,T]$. Moreover the following estimate holds: $\ \forall S\in \lbrack
0,T],$
\begin{equation}
\begin{array}{llll}
\Vert (Y(\cdot ),Z(\cdot ,\cdot ))\Vert _{\mathcal{H}^{2}[S,T]}^{2} & \equiv
& E\{ { \int_{S}^{T}}|Y(t)|^{2}dt+{ \int_{S}^{T}%
}{ \int_{S}^{T}}|Z(t,s)|^{2}dsdt\} &  \\
& \leq & CE\{ { \int_{S}^{T}}|\psi (t)|^{2}dt+E{%
 \int_{S}^{T}} ({ \int_{t}^{T}}|g_{0}(t,s)|ds )%
^{2}dt\} . &
\end{array}
\label{equ2.02}
\end{equation}
Let $\bar{g}:\Omega \times \lbrack 0,T]\times \lbrack 0,T]\times \mathbb{R}%
^{m}\times \mathbb{R}^{m\times d}\rightarrow \mathbb{R}^{m}$ also satisfies $%
\left( \mathbf{H}\right) $, $\psi(\cdot )\in L_{\mathcal{F}%
_{T}}^{2}(0,T)$ and $(\bar{Y}(\cdot ),\bar{Z}(\cdot ,\cdot ))\in \mathcal{H}%
^{2}[0,T]$ is the adapted $M$-solution of (\ref{equ2.1}) with $g$
and $\psi
(\cdot )$ replaced by $\bar{g}$ and $\psi(\cdot )$, respectively, then $%
\forall S\in \lbrack 0,T],$
\begin{eqnarray}
&&E\{ { \int_{S}^{T}}|Y(t)-\bar{Y}(t)|^{2}dt+{
\int_{S}^{T}}{ \int_{S}^{T}}|Z(t,s)-\bar{Z}(t,s)|^{2}dsdt
\} {}  \notag  \label{equ2.03} \\
&\leq &{}CE\{ { \int_{S}^{T}}|\psi (t)-\psi(t)|^{2}dt+%
{ \int_{S}^{T}} ({ \int_{t}^{T}}%
|g(t,s,Y(s),Z(s,t))-\bar{g}(t,s,Y(s),Z(s,t))|ds )^{2}dt\} .  \notag
\end{eqnarray}
\end{lemma}

Yong proved the following two duality principles for linear SVIE and
linear BSVIE in \cite{JY4, JY} respectively. And they play a key
role in deriving the maximum principle.

\begin{lemma}
\label{th11} Let $A_{i}(\cdot ,\cdot )\in L^{\infty }([0,T];L_{\mathbb{F}%
}^{\infty }(0,T;\mathbb{R}^{d\times d}))$ $(i=0,1\cdots d)$, $\varphi (\cdot
)\in L_{\mathbb{F}}^{2}(0,T;\mathbb{R}^{d})$, and $\psi (t)\in
L^{2}((0,T)\times \Omega ;\mathbb{R}^{d})$. Let $\xi (\cdot )\in L_{\mathbb{F%
}}^{2}(0,T;\mathbb{R}^{d})$ be the solution of the following FSVIE:
\begin{equation*}
\xi (t)=\varphi (t)+\int_{0}^{T}A_{0}(t,s)\xi (s)ds+{%
 \int_{0}^{t}}\sum_{i=1}^{d}A_{i}(t,s)\xi (s)dB_{i}(s),\quad
t\in \lbrack 0,T].
\end{equation*}
$(Y(\cdot ),Z(\cdot ,\cdot ))\in \mathcal{H}^{2}[0,T]$ be the adapted $M$%
-solution to the following BSVIE:
\begin{equation*}
Y(t)=\psi (t)+{ \int_{t}^{T}} [A_{0}(s,t)^{T}Y(s)+%
\sum_{i=1}^{d}A_{i}(s,t)^{T}Z_{i}(s,t) ]ds-{ \int_{t}^{T}}%
Z(t,s)dB_{s},\quad t\in \lbrack 0,T].
\end{equation*}%
Then the following relation holds:
\begin{equation*}
E\int_{0}^{T}\langle \xi (t),\psi (t)\rangle dt=E{%
 \int_{0}^{T}}\langle \varphi (t),Y(t)\rangle dt.
\end{equation*}
\end{lemma}

\begin{lemma}
\label{th12} Let $A_{i}(\cdot,\cdot)\in L^{\infty}([0,T];L^{\infty}_{\mathbb{%
F}}(0,T;\mathbb{R}^{d\times d}))$ $(i=0,1\cdots d)$, $\varphi(\cdot)\in
L^{2}_{\mathbb{F}}(0,T;\mathbb{R}^{d})$, and $\psi(t)\in
L^{2}((0,T)\times\Omega;\mathbb{R}^{d})$. Suppose $(Y(\cdot),Z(\cdot,\cdot))%
\in\mathcal{H}^2(0,T)$ is the solution of the following linear BSVIE:
\begin{equation*}
Y(t)=\psi(t)+{ \int_{t}^{T}} [A_{0}(t,s)^{T}Y(s)+%
\sum_{i=1}^dA_{i}(t,s)^{T}Z_{i}(s,t) ]ds-{ \int_{t}^{T}}%
Z(t,s)dB_s,\quad t\in[0,T],
\end{equation*}
and $X(\cdot)$ is the solution of the following FSVIE:
\begin{equation*}
X(t)=\varphi(t)+\int_{0}^{T}A_{0}(s,t)X(s)ds+{
\int_{0}^{t}}X(s)\sum_{i=1}^dE[A_{i}(s,t)|\mathcal{F}_s]dB_{i}(s),\quad t\in[%
0,T].
\end{equation*}
Then the following relation holds:
\begin{equation*}
E\int_{0}^{T}\langle X(t),\psi(t)\rangle dt=E{
\int_{0}^{T}}\langle \varphi(t),Y(t)\rangle dt.
\end{equation*}
\end{lemma}

For the proofs of Lemmas \ref{th11} and \ref{th12}, the readers are
referred to Theorem 5.1 in \cite{JY} and Theorem 3.1 in \cite{JY4},
respectively.

\section{Stochastic optimization problem}

\subsection{ One kind of stochastic optimization problem}

Let $K,\ \bar{K}$ be a nonempty convex subset of $\mathbb{R}^{m}$, set
\begin{equation*}
U[0,T]= \{u:[0,T]\times \Omega \rightarrow \mathbb{R}^{m} |u(\cdot )\in L_{\mathbb{F}}^{2}(0,T),\ u(s )\in K,\ s\in[0,T],\ a.e., a.s.%
 \},
\end{equation*}%
and
\begin{equation*}
\mathcal{U}=\left\{ (\psi(\cdot),u(\cdot))|\psi (\cdot )\in L_{\mathcal{F}%
_{T}}^{2}(0,T),\ \psi (t )\in \bar{K},\ t\in[0,T],\,\ a.e., a.s.,\
u(\cdot)\in U[0,T]\right\} .
\end{equation*}

For any given control pair $(\psi(\cdot), u(\cdot))\in
\mathcal{U},$ we consider the following controlled integral
equation:
\begin{equation}
\left\{
\begin{array}{llll}
X(t) =  f(t)+\int_{0}^{t}b(t,s,X(s),u(s))ds+{%
 \int_{0}^{t}}\sigma (t,s,X(s),u(s))dB_{s},\ \quad t\in
\lbrack 0,T], &  \\
Y(t)  =  \psi (t)+{ \int_{t}^{T}}%
g(t,s,X(s),Y(s),Z(s,t),u(s))ds-{ \int_{t}^{T}}Z(t,s)dB_{s}, &
\end{array}%
\right.  \label{equ2.2}
\end{equation}%
where $f(\cdot )\in L_{\mathbb{F}}^{2}(0,T)$ and
$
b:\Omega \times \Delta \lbrack 0,T]\times \mathbb{R}^{m}\times K\rightarrow
\mathbb{R}^{m},  \sigma :\Omega \times \Delta \lbrack 0,T]\times \mathbb{%
R}^{m}\times K\rightarrow \mathbb{R}^{m\times d},   \\
g:\Omega \times \Delta ^{c}[0,T]\times \mathbb{R}^{m}\times \mathbb{R}%
^{m}\times \mathbb{R}^{m\times d}\times K\rightarrow \mathbb{R}^{m}.
$

For each $(\psi(\cdot),u(\cdot))\in \mathcal{U}$, define the
following objective functional:
\begin{equation}
\begin{array}{llll}
J(\psi(\cdot),u(\cdot))  :=  E [\int_{0}^{T}{ \int_{t}^{T}%
}l_{2}(t,s,X(s),Y(s),Z(s,t),u(s))dsdt &  \\
 \qquad\qquad\qquad\qquad +\int_{0}^{T}{ \int_{0}^{t}%
}l_{1}(t,s,X(s),u(s))dsdt+\int_{0}^{T}q(\psi
(t))dt+h(X(T))+\int_0^T k(Y(t))dt ], &
\end{array}
\label{equcost}
\end{equation}%
where $l_{1}:\Delta \lbrack 0,T]\times \mathbb{R}^{m}\times K\rightarrow
\mathbb{R},\ \ l_{2}:\Delta ^{c}[0,T]\times \mathbb{R}^{m}\times \mathbb{R}%
^{m}\times \mathbb{R}^{m\times d}\times K\rightarrow \mathbb{R},\ q:\mathbb{R%
}^{m}\rightarrow \mathbb{R},\ h:\mathbb{R}^{m}\rightarrow \mathbb{R},\ k:%
\mathbb{R}^{m}\rightarrow \mathbb{R}.$ \newline We assume:
\begin{description}
\item[$\left( \mathbf{A_1}\right) $] $b,\ \sigma,\ g,\ l_1,\ l_2,\ q,\ h,\ k $ are continuous in their
argument, and continuously differentiable in the variables
$(x,y,\zeta,u)$; \end{description}
\begin{description}
\item[$\left( \mathbf{A_2}\right) $]  the derivatives of $b,\ \sigma,\ g,\ h$ in $(x,y,\zeta,u)$ are
bounded;\end{description}
\begin{description}
\item[$\left( \mathbf{A_3}\right) $]  the derivatives of $l_1, \ l_2$ in $(x,y,\zeta,u)$ are bounded by $%
C(1+|x|+|y|+|\zeta|+|u|)$, and the derivatives of $q,\ h,\ k$ in $x$
are bounded by $C(1+|x|)$;\end{description}
\begin{description}
\item[$\left( \mathbf{A_4}\right) $]  $g(t,s,x,y,\zeta ,u)$ is $\mathcal{F}_{T}\otimes \mathcal{B}(\Delta
^{c}[0,T]\times \mathbb{R}^{m}\times \mathbb{R}^{m}\times \mathbb{R}%
^{m\times d}\times K)$-measurable such that $s\mapsto g(t,s,x,y,\zeta ,u)$, $%
s\mapsto g_{i}(t,s,x,y,\zeta ,u)$ are $\mathbb{F}$-progressively measurable
for all $(t,x,y,\zeta ,u)\in \lbrack 0,T]\times \mathbb{R}^{m}\times \mathbb{%
R}^{m}\times \mathbb{R}^{m\times d}\times K$, $i=x,\ y,\ \zeta ,\ u$, and $E{%
 \int_{0}^{T}}( { \int_{t}^{T}}%
|g(t,s,0,0,0,0)|ds) ^{2}dt<\infty .$\end{description}

Under the assumptions $\left( \mathbf{A_1}\right) $, $\left(
\mathbf{A_2}\right)$ and $\left( \mathbf{A_4}\right) $, for any
given $u(\cdot)\in U[0,T]$, the FSVIE in (\ref{equ2.2}) has
a unique solution $X^{u}(\cdot )\in L_{\mathbb{F}}^{2}(0,T)$. For
any given $\psi (\cdot )\in L_{\mathcal {F}_T}^{2}(0,T)$, the BSVIE
has a unique M-solution $(Y^{\psi ,u}(\cdot ),Z^{\psi ,u}(\cdot
,\cdot ))\in \mathcal{H}^{2}[0,T]$ associated with $(\psi (\cdot
),u(\cdot))$. Hence, there exists a unique triple
$(X^{u}(\cdot ),Y^{\psi ,u}(\cdot ),Z^{\psi ,u}(\cdot ,\cdot ))$
satisfying (\ref{equ2.2}).

Now we formulate the optimization problem:
\begin{equation}
\begin{array}{l}
\mbox{Maximum}\hskip2.2cmJ(\psi(\cdot),u(\cdot)) \\
\mbox{subject to}\qquad (\psi(\cdot),u(\cdot))\in
\mathcal{U},\ \int_0^TEY^{\psi ,u}(s)ds=a,
\\
\qquad \qquad \qquad \ \ EY^{\psi,u}(t)=\rho (t),\ a.e.%
\end{array}
\label{equ2.0}
\end{equation}%
where $\rho :[0,T]\rightarrow \mathbb{R}^{m}$ is continuous and satisfies $\int_0^T\rho(t)dt=a$, ${
\int_{0}^{T}}|\rho (t)|^{2}dt<\infty .$

\subsection{Variational equation}

For $(\psi ^{1}(\cdot),u^{1}(\cdot)),(\psi
^{2}(\cdot),u^{2}(\cdot))\in \mathcal{U}$, we define a metric
in $\mathcal{U}$ by
\begin{equation*}\begin{array}{llll}
d((\psi ^{1}(\cdot),u^{1}(\cdot)),(\psi
^{2}(\cdot),u^{2}(\cdot))):=( E \int_{0}^{T}|\psi
^{1}(s)-\psi ^{2}(s)|^{2}ds+E\int_{0}^{T}|u^{1}(s)-u^{2}(s)|^{2}ds) ^{\frac{1}{2}}.
\end{array}\end{equation*}%
It is obvious that $(\mathcal{U},d(\cdot ,\cdot ))$ is a complete metric
space.

Let $(\psi ^{\ast }(\cdot),u^{\ast }(\cdot))$ be an optimal control pair to problem (\ref%
{equ2.0}) and $(X^{\ast }(\cdot ),Y^{\ast }(\cdot ),Z^{\ast }(\cdot
,\cdot
)) $ be the corresponding state processes of (\ref{equ2.2}).  For any $%
(\psi(\cdot),u(\cdot)\in \mathcal {U},\ 0\leq p\leq 1,$ using the convexity of
 $\mathcal{U}$, we have
\begin{equation*}
\begin{array}{llll}
(\psi ^{p}(\cdot),u^{p}(\cdot))  :=  ((1-p)\psi ^{\ast }(\cdot)+p\psi(\cdot),(1-p)u^{\ast }(\cdot)+pu(\cdot)) &  \\
\qquad\qquad\qquad = (\psi ^{\ast
}(\cdot)+p(\psi(\cdot)-\psi^{\ast }(\cdot)),u^{\ast
}(\cdot)+p(u(\cdot)-u^{\ast }(\cdot)))\in\mathcal
{U}. &
\end{array}%
\end{equation*}%
We denote $(X^{p}(\cdot ),Y^{p}(\cdot ),Z^{p}(\cdot ,\cdot ))$ by
the solution of the corresponding FBSVIE (\ref{equ2.2}) with
$(\psi(\cdot),$ $u(\cdot))=(\psi
^{p}(\cdot),u^{p}(\cdot))$.

Consider the following FBSVIE
\begin{equation}  \label{equ2.3}
\left\{
\begin{array}{llll}
\delta X(t) =  \int_{0}^{t}b^\ast_{u}(t,s)\hat{u}(s)ds +{%
 \int_{0}^{t}}\sigma^\ast_{u}(t,s) \hat{u}(s)dB_s &  \\
\qquad\qquad+\int_{0}^{T}b^\ast_{x}(t,s)\delta X(s)ds +{
\int_{0}^{t}}\sigma^\ast_{x}(t,s)\delta X(s)dB_s,\qquad t\in[0,T], &  \\
\delta Y(t)  =  \hat{\psi}(t)+{ \int_{t}^{T}}[%
g^\ast_{x}(t,s)\delta X(s) +g^\ast_{y}(t,s)\delta Y(s)
+g^\ast_{\zeta}(t,s)\delta Z(s,t) &  \\
 \qquad\qquad +g^\ast_{u}(t,s)\hat{u}(s)]ds -{ \int_{t}^{T}}\delta
Z(t,s)dB_s, &
\end{array}%
\right.
\end{equation}
where $\hat{\psi}(s)=\psi(s)-\psi^\ast(s),\ \hat{u}(s)=u(s)-u^\ast(s),\
f^\ast_k(t,s)=f_{k}(t,s,X^{\ast}(s),Y^{\ast}(s),Z^\ast(s,t),u^{\ast}(s)),\
k=x,\ y,\ \zeta,\ u,$ $f=b,\ \sigma,\ g,$ respectively. This
equation is called the variational equation.

From Lemma \ref{th1} and $\left(\mathbf{A}_{1}\right),\ \left(\mathbf{A}%
_{2}\right),\ \left(\mathbf{A}_{4}\right)$, it's easy to check that the
variational equation (\ref{equ2.3}) has a unique solution $(\delta
X(\cdot),\delta Y(\cdot),\delta Z(\cdot,\cdot))\in L^2_\mathbb{F}(0,T)\times
\mathcal{H}^2[0,T].$

Now we define
\begin{equation*}
\begin{array}{l}
\tilde{X}^{p}(t)=p^{-1}[X^{p}(t)-X^{\ast }(t)]-\delta X(t), \\
\tilde{Y}^{p}(t)=p^{-1}[Y^{p}(t)-Y^{\ast }(t)]-\delta Y(t), \\
\tilde{Z}^{p}(t,s)=p^{-1}[Z^{p}(t,s)-Z^{\ast }(t,s)]-\delta Z(t,s).%
\end{array}%
\end{equation*}%
To simplify the proof, we use the following notations:
\begin{equation*}
\begin{array}{l}
f^{p}(t,s)=f(t,s,X^{p}(s),Y^{p}(s),Z^{p}(s,t),u^{p}(s)), \\
f^{\ast }(t,s)=f(t,s,X^{\ast }(s),Y^{\ast }(s),Z^{\ast }(s,t),u^{\ast}(s)),%
\end{array}%
\end{equation*}%
where $f=b, \sigma, g,$ respectively. Similar to the arguments in
\cite {SJSP, WS}, we have the following lemma:

\begin{lemma}
\label{lem2} Assume that $\left( \mathbf{A}_{1}\right) ,\ \left( \mathbf{A}%
_{2}\right) ,\ \left( \mathbf{A}_{4}\right) $ hold. We have
\begin{equation*}
\begin{array}{l}
\lim\limits_{p\rightarrow 0}E\int_{0}^{T}|\tilde{X}%
^{p}(t)|^{2}dt=0, \
\lim\limits_{p\rightarrow 0}E\int_{0}^{T}|\tilde{Y}%
^{p}(t)|^{2}dt=0, \\
\lim\limits_{p\rightarrow 0}E\int_{0}^{T}{
\int_{0}^{T}}|\tilde{Z}^{p}(t,s)|^{2}dsdt=0.%
\end{array}%
\end{equation*}
\end{lemma}

\begin{proof}
\textbf{(1).} We prove the first equality. By the FSVIEs in
(\ref{equ2.2}) and (\ref{equ2.3}), we have
\begin{equation*}
\begin{array}{ll}
\tilde{X}^{p}(t) = \int_{0}^{t}\frac{1}{p}[b^p(t,s)
{} -b^\ast(t,s) -pb^\ast_{x}(t,s)\delta X(s)
-pb^\ast_{u}(t,s)\hat{u}(s)]ds {}
\nonumber\\
\qquad\qquad+ \int_{0}^{t}\frac{1}{p}[\sigma^p(t,s)
{} -\sigma^\ast(t,s) -p\sigma^\ast_{x}(t,s)\delta X(s)
-p\sigma^\ast_{u}(t,s)\hat{u}(s)]dB_s {} \nonumber\\
\qquad\ \
=\int_{0}^{t}[A^p_1(t,s)\tilde{X}^{p}(s)+D_1^p(t,s)]ds+\int_{0}^{t}[A^p_2(t,s)\tilde{X}^{p}(s)+D_2^p(t,s)]dB_s,
\end{array}\end{equation*} where
$$
\begin{array}
[c]{l}%
A^{p}_1(t,s):={ \int_{0}^{1}}b_{x}(t,s,L(p, \lambda,s),M(p,\lambda,s))d\lambda,\ A^{p}_2(t,s):={ \int_{0}^{1}}\sigma_{x}(t,s,L(p,\lambda,s),M(p,\lambda,s))d\lambda,\\
B^{p}_1(t,s):={ \int_{0}^{1}}b_{u}(t,s,L(p,\lambda,s),M(p,\lambda,s))d\lambda,\ B^{p}_2(t,s):={ \int_{0}^{1}}\sigma_{u}(t,s,L(p,\lambda,s),M(p,\lambda,s))d\lambda,\\
D^{p}_1(t,s):=[A^{p}_1(t,s)-b^{\ast}_{x}(t,s)]\delta
X(s)+[B^{p}_1(t,s)-b^{\ast}_{u}(t,s)]\hat{u}(s),\\
D^{p}_2(t,s):=[A^{p}_2(t,s)-\sigma^{\ast}_{x}(t,s)]\delta
X(s)+[B^{p}_2(t,s)-\sigma^{\ast}_{u}(t,s)]\hat{u}(s),
\end{array}$$ and
$$
\begin{array}
[c]{l}%
L(p,\lambda,s):=X^{\ast}(s)+\lambda (X^{p}(s)- X^\ast(s)),\\
M(p,\lambda,s):=u^{\ast}(s)+\lambda (u^{p}(s)-u^{\ast}(s)).
\end{array}
$$
Therefore, we have $$
 \begin{array}{llll}
E\int_{0}^{T}e^{-rt}|\tilde{X}^{p}(t)|^2dt &\\
 \leq
CE\int_{0}^{T}e^{-rt}\int_{0}^{T}(|A_1^p(t,s)|^2+|A_2^p(t,s)|^2)|\tilde{X}^{p}(s)|^2dsdt &\\
\quad+CE\int_{0}^{T}e^{-rt}\int_{0}^{T}(|D_1^p(t,s)|^2+|D_2^p(t,s)|^2)dsdt &\\
 \leq \frac{C}{r}E\int_{0}^{T}e^{-rt}|\tilde{X}^{p}(t)|^2dt+CE\int_{0}^{T}\int_{0}^{T}e^{-rt}(|D_1^p(t,s)|^2+|D_2^p(t,s)|^2)dsdt.\\
\end{array}$$
By choosing a proper $r$ such that $\frac{C}{r}<1$, we have
$$E\int_{0}^{T}|\tilde{X}^{p}(t)|^2dt\\
\leq
CE\int_{0}^{T}\int_{0}^{T}e^{r(T-t)}(|D_1^p(t,s)|^2+|D_2^p(t,s)|^2)dsdt.$$
 Applying Lebesgue's dominated convergence theorem, we have
$$\lim\limits_{p\rightarrow0}E\int_{0}^{T}\int_{0}^{T}|D^{p}_i(t,s)|^{2}dsdt=
0,\ i=1,2.$$ So,
$$\lim\limits_{p\rightarrow0}E\int_{0}^{T}|\tilde{X}^{p}(t)|^2dt=0.$$

\textbf{(2).}  By the BSVIEs in (\ref{equ2.2}) and
(\ref{equ2.3}), we have
\begin{eqnarray*}\begin{array}{llll}
\tilde{Y}^{p}(t) = { \int_{t}^{T}}\frac{1}{p}[g^p(t,s)
{} -g^\ast(t,s) -pg^\ast_{x}(t,s)\delta X(s)-pg^\ast_{y}(t,s)\delta
Y(s) -pg^\ast_{\zeta}(t,s)\delta Z(s,t)& {}
\nonumber\\
\qquad\qquad-pg^\ast_{u}(t,s)\hat{u}(s)]ds
-{ \int_{t}^{T}}\tilde{Z}^{p}(t,s)dB_s,\quad t\in[0,T].
\end{array}\end{eqnarray*} Let
$$
\begin{array}
[c]{l}%
N(p,\lambda,s):=Y^{\ast}(s)+\lambda (Y^{p}(s)- Y^\ast(s)),\\
P(p,\lambda,t,s):=Z^{\ast}(s,t)+\lambda (Z^{p}(s,t)- Z^\ast(s,t))\\
\end{array}
$$
and
$$
\begin{array}
[c]{l}%
C^{p}_1(t,s):={ \int_{0}^{1}}g_{x}(t,s,L(p,\lambda,s),N(p,\lambda,s),P(p,\lambda,t,s),M(p,\lambda,s))d\lambda,\\
C^{p}_2(t,s):={ \int_{0}^{1}}g_{y}(t,s,L(p,\lambda,s),N(p,\lambda,s),P(p,\lambda,t,s),M(p,\lambda,s))d\lambda,\\
C^{p}_3(t,s):={ \int_{0}^{1}}g_{\zeta}(t,s,L(p,\lambda,s),N(p,\lambda,s),P(p,\lambda,t,s),M(p,\lambda,s))d\lambda,\\
C^{p}_4(t,s):={ \int_{0}^{1}}g_{u}(t,s,L(p,\lambda,s),N(p,\lambda,s),P(p,\lambda,t,s),M(p,\lambda,s))d\lambda,\\
D^{p}(t,s):=[C^{p}_1(t,s)-g^{\ast}_{x}(t,s)]\delta
X(s)+[C^{p}_2(t,s)-g^{\ast}_{y}(t,s)]\delta
Y(s)\\
\qquad\qquad+[C^{p}_3(t,s)-g^{\ast}_{\zeta}(t,s)]\delta
Z(s,t)+[C^{p}_4(t,s)-g^{\ast}_{u}(t,s)]\hat{u}(s).
\end{array}
$$
Thus,
\begin{eqnarray*}\begin{array}{llll}
\tilde{Y}^{p}(t) =
{ \int_{t}^{T}}[C^{p}_1(t,s)\tilde{X}^{p}(s)+C^{p}_2(t,s)\tilde{Y}^{p}(s)
+C^{p}_3(t,s)\tilde{Z}^{p}(s,t)+D^{p}(t,s)]ds &{}
\nonumber\\
 \qquad\qquad-{ \int_{t}^{T}}\tilde{Z}^{p}(t,s)dB_s, \quad
t\in[0,T]. \end{array}\end{eqnarray*}
In Lemma \ref{th1}, we take
$\psi=0,\ g_{0}(t,s)=C^{p}_1(t,s)\tilde{X}^{p}(s)+D^{p}(t,s). $ Then
\begin{eqnarray*}\begin{array}{llll}\|\tilde{Y}^{p}(t),\tilde{Z}^{p}(t,s))\|_{\mathcal {H}^{2}[0,T]}^{2}
 =  E[\int_{0}^{T}|\tilde{Y}^{p}(t)|^{2}ds+
\int_{0}^{T}\int_{0}^{T}|\tilde{Z}^{p}(t,s)|^{2}dsdt
] & {}
\nonumber\\
  \qquad \qquad \qquad \qquad \qquad\ \leq
CE\int_{0}^{T}[({ \int_{t}^{T}}|C_1^{p}(t,s)\tilde{X}^{p}(s)|ds)^{2}+({ \int_{t}^{T}}|D^{p}(t,s)|ds)^{2}]dt.
\end{array}\end{eqnarray*} Applying Lebesgue's dominated convergence
theorem, we have
$$\lim\limits_{p\rightarrow0}E\int_{0}^{T}({ \int_{t}^{T}}|D^{p}(t,s)|ds)^{2}dt\rightarrow
0.$$ Using the obtained first result, we can get the desired
results.
\end{proof}

\subsection{Variational inequality}

In this subsection, using Ekeland's variational principle (\cite{IE}), we get
the variational inequality.

\begin{lemma}[Ekeland's variational principle]
\label{lem1} Let $(V,d(\cdot,\cdot))$ be a complete metric space and $%
F(\cdot): V\rightarrow R$ be a proper lower semi-continuous function bounded
from below. Suppose that for some $\varepsilon> 0$, there exists $u \in V$
satisfying $F(u)\leq \inf\limits_{v\in V}F(v)+\varepsilon. $ Then there
exists $u_{\varepsilon}\in V$ such that
\begin{equation*}
\begin{array}{ll}
\rm(i)\quad & F(u_{\varepsilon})\leq F(u), \\
\rm(ii)\quad & d(u,u_{\varepsilon})\leq\varepsilon, \\
\rm(iii)\quad & F(v)+\sqrt{\varepsilon}d(v,u_{\varepsilon})\ge
F(u_{\varepsilon}), \quad\forall v\in V.%
\end{array}%
\end{equation*}
\end{lemma}

Given the optimal control pair $(\psi ^{\ast }(\cdot),u^{\ast
}(\cdot))\in \mathcal{U},$ introduce a
mapping $F_{\varepsilon }(\cdot ):\mathcal{U}\rightarrow \mathbb{R}$ by%
\newline
$%
\begin{array}{llll}
F_{\varepsilon }(\psi(\cdot) ,u(\cdot))  :=  \{ |\int_0^TEY(t)dt- a |^{2}+{
\int_{0}^{T}}|EY(t)-\rho (t)|^{2}dt&\\
 \qquad \qquad+\{\max(0, \int_0^TEk(Y^\ast(s))ds-\int_0^TEk(Y(s))ds+\varepsilon)\}^2&\\
 \qquad \qquad+\{\max (0,Eh(X^{\ast }(T))-Eh(X(T))+\varepsilon )%
\}^{2}+\{\max(0,\int_{0}^{T}Eq(\psi ^{\ast
}(t))dt-\int_{0}^{T}Eq(\psi (t))dt+\varepsilon )\}%
^{2} &  \\
 \qquad \qquad +\{\max (0,\int_{0}^{T}{
\int_{0}^{t}}El_{1}^{\ast }(t,s)dsdt-\int_{0}^{T}{
\int_{t}^{T}}El_{1}(t,s)dsdt+\varepsilon )\}^{2} &  \\
  \qquad \qquad+\{\max (0,\int_{0}^{T}{
\int_{t}^{T}}El_{2}^{\ast }(t,s)dsdt-\int_{0}^{T}{
\int_{t}^{T}}El_{2}(t,s)dsdt+\varepsilon)\}^{2} \} ^{1/2}, &
\end{array}%
$\newline where $ l_{i}^{\ast }(t,s)=l_{i}(t,s,X^{\ast }(s),Y^{\ast
}(s),Z^{\ast }(s,t),u^{\ast }(s)),\
l_{i}(t,s)=l_{i}(t,s,X(s),Y(s),Z(s,t),u(s)),$ $i=1,\ 2, $
$\varepsilon $ is an arbitrary positive constant, $l_{i},\ q,\ h,\
k$ satisfy $\left( \mathbf{A}_{1}\right) ,\ \left(
\mathbf{A}_{2}\right) ,\ \left( \mathbf{A}_{3}\right) $.

\begin{remark}
Under $\left( \mathbf{A}_{1}\right) -\left( \mathbf{A}_{4}\right) $, from
the wellposedness of BSVIEs (Lemma \ref{th1}) as well as the proof of Lemma %
\ref{lem2}, we know that $F_{\varepsilon }(\cdot,\cdot)$ is a continuous function on $%
\mathcal{U}$.
\end{remark}

\begin{theorem}
\label{th3} Let $(\psi ^{\ast }(\cdot),u^{\ast }(\cdot))\in
\mathcal{U}$ be the optimal
control pair. Under the assumptions $\left( \mathbf{A}_{1}\right) -\left( \mathbf{%
A}_{4}\right) $, there exist a deterministic function $h_{0}(\cdot )\in
\mathbb{R}^{m},\ \bar{h}_{0}\in \mathbb{R}^{m}$, $\bar{h}_{1},\ h_{1},\
h_{2},\ h_{3},\ h_{4}\in \mathbb{R},\ h_{1},\ h_{2},\ h_{3},\ h_{4}$ $\leq
0,\ |\bar{h}_{0}|+|h_{0}(\cdot )|+|\bar{h}%
_{1}|+|h_{1}|+|h_{2}|+|h_{3}|+|h_{4}|\neq 0$ such that the following
variational inequality holds
\begin{equation*}
\begin{array}{llll}
&  & \int_{0}^{T}E\langle h_{0}(t)+\bar{h}_{0},\delta Y(t)\rangle dt+\bar{%
h}_{1}\int_{0}^{T}E\langle q_{x}(\psi ^{\ast }(t)),\hat{\psi}%
(t)\rangle dt+h_{1}E\langle h_{x}(X^{\ast }(T)),\delta X(T)\rangle &  \\
&  & +h_{2}\int_0^TE\langle k_{y}(Y^{\ast }(s)),\delta Y(s)\rangle ds +h_{3}{%
\int_{0}^{T}}\int_{0}^{t}E\langle l_{1x}^{\ast
}(t,s),\delta X(s)\rangle dsdt &  \\
&  & +h_{3}\int_{0}^{T}\int_{0}^{t}E\langle
l_{1u}^{\ast }(t,s)^{T},\hat{u}^{\ast }(s)\rangle dsdt+h_{4}{
\int_{0}^{T}}{\int_{t}^{T}}E\langle l_{2x}^{\ast }(t,s),\delta
X(s)\rangle dsdt &  \\
&  & +h_{4}\int_{0}^{T}{\int_{t}^{T}}E\langle
l_{2y}^{\ast }(t,s),\delta Y(s)\rangle dsdt+h_{4}\int_{0}^{T}{\int_{t}^{T}}E\langle l_{2\zeta }^{\ast }(t,s),\delta
Z(s,t)\rangle dsdt &  \\
&  & +h_{4}\int_{0}^{T}{\int_{t}^{T}}E\langle
l_{2u}^{\ast }(t,s),\hat{u}^{\ast }(s)\rangle dsdt\geq 0,
\end{array}%
\end{equation*}%
where $l_{ik}^{\ast }(t,s),\ i=1,\ 2,\ k=x,\ y,\ \zeta ,\ u,$ is the
derivative of $l_{i}^{\ast }(t,s)$ with respect to $k$, respectively.
\end{theorem}

\begin{proof} It is easy to check that the following properties hold
$$
\begin{array}
[c]{ll}%
\rm(i)\quad & F_{\varepsilon}(\psi^{\ast}(\cdot),u^{\ast}(\cdot))=\sqrt{2}\varepsilon,\\
\rm(ii)\quad & F_{\varepsilon}(\psi(\cdot),u(\cdot))>0, \quad\forall (\psi(\cdot),u(\cdot))\in \mathcal{U},\\
\rm(iii)\quad &
F_{\varepsilon}(\psi^{\ast}(\cdot),u^{\ast}(\cdot))\leq
\inf\limits_{(\psi,u)\in
\mathcal{U}}F_{\varepsilon}(\psi(\cdot),u(\cdot))+\sqrt{2}\varepsilon.
\end{array}
$$
Then from Lemma \ref{lem1}~(Ekeland's variational principle), we can
find a $(\psi^{\varepsilon}(\cdot),u^{\varepsilon}(\cdot))\in
\mathcal{U}$, such that
$$
\begin{array}
[c]{ll}%
\rm(i)\quad & F_{\varepsilon}(\psi^{\varepsilon}(\cdot),u^{\varepsilon}(\cdot))\leq F_{\varepsilon}(\psi^{\ast}(\cdot),u^{\ast}(\cdot)),\\
\rm(ii)\quad & d((\psi^{\varepsilon}(\cdot),u^{\varepsilon}(\cdot)),(\psi^{\ast}(\cdot),u^{\ast}(\cdot)))\leq\sqrt{2}\varepsilon,\\
\rm(iii)\quad &
F_{\varepsilon}(\psi(\cdot),u(\cdot))+\sqrt{\sqrt{2}\varepsilon}
d((\psi(\cdot),u(\cdot)),(\psi^{\varepsilon}(\cdot),u^{\varepsilon}(\cdot)))\ge
F_{\varepsilon}(\psi^{\varepsilon}(\cdot),u^{\varepsilon}(\cdot)),\\
&\hskip10cm \forall (\psi(\cdot),u(\cdot))\in \mathcal{U}.
\end{array}
$$
 For each $(\psi(\cdot),u(\cdot))\in \mathcal{U}$, we define
$$
(\hat{\psi}(\cdot),\hat{u}(\cdot)):=(\psi(\cdot)-\psi^{\ast}(\cdot),u(\cdot)-u^{\ast}(\cdot)),\
(\hat{\psi}^{\varepsilon}(\cdot),\hat{u}^{\varepsilon}(\cdot)):=(\psi(\cdot)-\psi^{\varepsilon}(\cdot),u(\cdot)-u^{\varepsilon}(\cdot)),$$
then
$(\psi_{p}^{\varepsilon}(\cdot),u_{p}^{\varepsilon}(\cdot)):=(\psi^{\varepsilon}(\cdot)+p\hat{\psi}^{\varepsilon}(\cdot),u^{\varepsilon}(\cdot)+p\hat{u}^{\varepsilon}(\cdot))\in
U$. Indeed,
$(\psi^{\varepsilon}(\cdot),u^{\varepsilon}(\cdot))\in
\mathcal{U}$,
$(\hat{\psi}^{\varepsilon}(\cdot)+\psi^{\varepsilon}(\cdot),\hat{u}^{\varepsilon}(\cdot)+u^{\varepsilon}(\cdot))=(\psi(\cdot),u(\cdot))\in
\mathcal{U}$, then
$$\begin{array}{llll}(\psi_{p}^{\varepsilon},u_{p}^{\varepsilon}):=(\psi^{\varepsilon}(\cdot)+p\hat{\psi}^{\varepsilon}(\cdot),u^{\varepsilon}(\cdot)+p\hat{u}^{\varepsilon}(\cdot))&\\
\qquad\quad\ \ \
=((1-p)\psi^{\varepsilon}(\cdot)+p(\hat{\psi}^{\varepsilon}(\cdot)+\psi^{\varepsilon}(\cdot)),(1-p)u^{\varepsilon}(\cdot)+p(\hat{u}^{\varepsilon}(\cdot)+u^{\varepsilon}(\cdot)))\in\mathcal{U}.
\end{array}$$
 Let
$(X_{p}^{\varepsilon}(\cdot),Y_{p}^{\varepsilon}(\cdot),Z_{p}^{\varepsilon}(\cdot,\cdot))$
(resp.
$(X^{\varepsilon}(\cdot),Y^{\varepsilon}(\cdot),Z^{\varepsilon}(\cdot,\cdot))$)
be the solution of BSVIE (\ref{equ2.2}) with $(\psi(\cdot),$ $
u(\cdot))=(\psi_{p}^{\varepsilon}(\cdot),u_{p}^{\varepsilon}(\cdot))$~(resp.
$(\psi(\cdot),
u(\cdot))=(\psi^{\varepsilon}(\cdot),u^{\varepsilon}(\cdot)))$.
From Ekeland's variational principle, it follows that
\begin{equation}\label{equ2.07}
F_{\varepsilon}(\psi_{p}^{\varepsilon}(\cdot),u_{p}^{\varepsilon}(\cdot))+\sqrt{\sqrt{2}\varepsilon}d((\psi_{p}^{\varepsilon}(\cdot),u_{p}^{\varepsilon}(\cdot)),(\psi^{\varepsilon}(\cdot),u^{\varepsilon}(\cdot)))-
F_{\varepsilon}(\psi^{\varepsilon}(\cdot),u^{\varepsilon}(\cdot))\ge0.
\end{equation}
We consider the following variational equation:
\begin{equation}\label{equ2.04}\left\{\begin{array}{ll}
\delta X^{\varepsilon}(t) =
\int_{0}^{t}b^\varepsilon_u(t,s)\hat{u}^\varepsilon(s)ds+\int_{0}^{t}\sigma^\varepsilon_u(t,s)\hat{u}^\varepsilon(s)dB_s &\\
 \qquad \qquad+\int_{0}^{t}b^\varepsilon_x(t,s)\delta
X^{\varepsilon}(s)ds+\int_{0}^{t}\sigma^\varepsilon_x(t,s)\delta
X^{\varepsilon}(s)
dB_s ,\quad t\in[0,T], &\\
\delta Y^{\varepsilon}(t)=
\hat{\psi}^{\varepsilon}(t)+{ \int_{t}^{T}}[g^\varepsilon_x(t,s)\delta
X^{\varepsilon}(s) +g^\varepsilon_y(t,s)\delta Y^{\varepsilon}(s)
 +g^\varepsilon_\zeta(t,s)\delta Z^{\varepsilon}(s,t)&\\
\qquad \qquad
+g^\varepsilon_u(t,s)(u(s)-u^\varepsilon(s))]ds-{ \int_{t}^{T}}\delta
Z^{\varepsilon}(t,s)dB_s,\end{array}\right.
\end{equation}
where
$f^\varepsilon_k(t,s)=f_{k}(t,s,X^{\varepsilon}(s),Y^{\varepsilon}(s),Z^{\varepsilon}(s,t),u^{\varepsilon}(s)),\
k=x,\ y,\ \zeta,\ u,\ f=b,\ \sigma,\ g,\  \mbox{respectively}.$

Similarly to Lemma \ref{lem2}, we have
$$\lim\limits_{p\rightarrow0}E\int_{0}^{T}|\frac{X^{\varepsilon}_{p}(t)-X^{\varepsilon}(t)}{p}
-\delta X^{\varepsilon}(t)|^{2}dt=0, \
\lim\limits_{p\rightarrow0}E\int_{0}^{T}|\frac{Y^{\varepsilon}_{p}(t)-Y^{\varepsilon}(t)}{p}
-\delta Y^{\varepsilon}(t)|^{2}dt=0, $$ which lead to the
following expansions:
$$\begin{array}{ll}EX^{\varepsilon}_{p}(t)-EX^{\varepsilon}(t)=pE\delta
X^{\varepsilon}(t)+o(p), &\\
EY^{\varepsilon}_{p}(t)-EY^{\varepsilon}(t)=pE\delta
Y^{\varepsilon}(t)+o(p),&\\
\int_{0}^{T} |EY^{\varepsilon}_{p}(t)-\rho(t) |^{2}dt-\int_{0}^{T} |EY^{\varepsilon}(t)-\rho(t) |^{2}dt&\\
= \int_{0}^{T}2p\langle
EY^{\varepsilon}(t)-\rho(t),E\delta Y^{\varepsilon}(t)\rangle
dt+o(p). \end{array}$$
 From $\left(
\mathbf{A}_{1}\right)$, we have $$\begin{array}
[c]{l}%
\int_{0}^{T}Eq(\psi_p^\varepsilon(t))dt-\int_{0}^{T}Eq(\psi^\varepsilon(t))dt=p\int_{0}^{T}E\langle q_x(\psi^\varepsilon(t)),\hat{\psi}^\varepsilon(t)\rangle dt+o(p),\\
Eh(X^{\varepsilon}_{p}(T))-Eh(X^{\varepsilon}(T))
=pE\langle h_{x}(X^{\varepsilon}(T)), \delta X^\varepsilon(T)\rangle+o(p),\\
\int_{0}^{T}Ek(Y^{\varepsilon}_{p}(t))dt-\int_{0}^{T}Ek(Y^{\varepsilon}(t))dt=p\int_0^T\langle Ek_y(Y^\varepsilon(s)),\delta Y^\varepsilon(s)  \rangle ds+o(p),\\
\int_{0}^{T}\int_{0}^{t}El_1^{p\varepsilon}(t,s)dsdt-\int_{0}^{T}\int_{0}^{t}El_1^\varepsilon(t,s)dsdt
=p\int_{0}^{T}\int_{0}^{t}E(\langle 1_{1x}^\varepsilon(t,s),\delta X^\varepsilon(s)\rangle +\langle 1_{1u}^\varepsilon(t,s), \hat{u}^\varepsilon(s)\rangle )dsdt+o(p),\\
\int_{0}^{T}\int_{t}^{T}El_2^{p\varepsilon}(t,s)dsdt-\int_{0}^{T}\int_{0}^{t}El_2^\varepsilon(t,s)dsdt\\
=p\int_{0}^{T}\int_{t}^{T}E(\langle 1_{2x}^\varepsilon(t,s),\delta X^\varepsilon(s)\rangle +\langle 1_{2y}^\varepsilon(t,s),\delta Y^\varepsilon(s)\rangle+\langle 1_{2z}^\varepsilon(t,s),\delta Z^\varepsilon(s,t)\rangle+\langle 1_{2u}^\varepsilon(t,s), \hat{u}^\varepsilon(s)\rangle) dsdt+o(p),
\end{array}$$
Furthermore, the following expansions hold:
$$\begin{array}{llll}
&&
|\int_0^TEY^{\varepsilon}_{p}(t)dt-a|^2-|\int_0^TEY^{\varepsilon}(t)dt-a)|^2 \\
 &=&2\langle\int_0^TEY^{\varepsilon}(t)dt-a),\int_0^TEY^{\varepsilon}_p(t)dt-\int_0^TEY^{\varepsilon}(t)dt\rangle+o(p)\\
 &=&2p[\int_0^TEY^{\varepsilon}(t)dt-a]\int_0^TE\delta Y^{\varepsilon}(t)dt+o(p),\\
&&[\int_{0}^{T}Eq(\psi^\ast(t))dt-\int_{0}^{T}Eq(\psi_p^\varepsilon(t))dt+\varepsilon]^{2}-[\int_{0}^{T}Eq(\psi^\ast(t))dt-\int_{0}^{T}Eq(\psi^\varepsilon(t))dt+\varepsilon]^{2}\\
&=&2\langle \int_{0}^{T}q(\psi^\varepsilon(t))dt  -\int_{0}^{T}q(\psi_p^\varepsilon(t))dt, \int_{0}^{T}Eq(\psi^\ast(t))dt-\int_{0}^{T}Eq(\psi^\varepsilon(t))dt+\varepsilon\rangle +o(p)      \\
&=&-2p[\int_{0}^{T}Eq(\psi^\ast(t))dt-\int_{0}^{T}Eq(\psi^\varepsilon(t))dt+\varepsilon]\int_{0}^{T}E\langle
q_x(\psi^\varepsilon(t)),\hat{\psi}^\varepsilon(t)\rangle dt+o(p),\\
&&[Eh(X^\ast(T)-Eh(X^{\varepsilon}_{p}(T)))+\varepsilon]^{2}
-[Eh(X^\ast(T)-Eh(X^{\varepsilon}(T)))+\varepsilon]^{2} {}
\nonumber\\
& = & {}2\langle Eh(X^{\varepsilon}(T))-
Eh(X^{\varepsilon}_{p}(T)),Eh(X^\ast(T))-Eh(X^{\varepsilon}(T))
+\varepsilon\rangle +o(p) {}
\nonumber\\
& = & {} -2p[Eh(X^\ast(T))-Eh(X^{\varepsilon}(T))
+\varepsilon]E\langle h_x(X^\varepsilon(T)),\delta
X^\varepsilon(T)\rangle+o(p),\\
 &&[\int_0^TEk(Y^\ast(s))ds-\int_0^T Ek(Y^{\varepsilon}_{p}(s))ds+\varepsilon]^{2}
-[ \int_0^T Ek(Y^\ast(s) )ds- \int_0^T Ek(Y^{\varepsilon}(s))ds+\varepsilon]^{2}
{}
\nonumber\\
& = & {} -2\langle  \int_0^T Ek(Y^{\varepsilon}(s))ds -
\int_0^T Ek(Y_p^{\varepsilon}(s))ds,  \int_0^T Ek(Y^\ast(s) )ds- \int_0^T Ek(Y^{\varepsilon}(s))ds+\varepsilon \rangle +o(p) {}
\nonumber\\
& = & {} -2p[ \int_0^T Ek(Y^\ast(s) )ds- \int_0^T Ek(Y^{\varepsilon}(s))ds+\varepsilon]\int_0^TE \langle
k_y(Y^{\varepsilon}(s)),\delta Y^\varepsilon(s)\rangle ds+o(p),\nonumber\\
&&(\int_{0}^{T}\int_{0}^{t}El_1^\ast(t,s)dsdt-\int_{0}^{T}\int_{0}^{t}El_1^{p\varepsilon}(t,s)dsdt+\varepsilon)^2-(\int_{0}^{T}\int_{0}^{t}El_1^\ast(t,s)dsdt-\int_{0}^{T}\int_{0}^{t}El_1^\varepsilon(t,s)dsdt+\varepsilon)^2\\
&=&2\langle \int_{0}^{T}\int_{0}^{t}El_1^{\varepsilon}(t,s)dsdt-\int_{0}^{T}\int_{0}^{t}El_1^{p\varepsilon}(t,s)dsdt,\int_{0}^{T}\int_{0}^{t}El_1^\ast(t,s)dsdt-\int_{0}^{T}\int_{0}^{t}El_1^\varepsilon(t,s)dsdt+\varepsilon\rangle+o(p)\\
&=&-2p[\int_{0}^{T}\int_{0}^{t}El_1^\ast(t,s)dsdt-\int_{0}^{T}\int_{0}^{t}El_1^\varepsilon(t,s)dsdt+\varepsilon][\int_{0}^{T}\int_{0}^{t}E(\langle l_{1x}^\varepsilon(t,s),\delta X^\varepsilon(s)\rangle +\langle l_{1u}^\varepsilon(t,s), \hat{u}^\varepsilon(s)\rangle)dsdt]+o(p),\\
&&(\int_{0}^{T}\int_{t}^{T}El_2^\ast(t,s)dsdt-\int_{0}^{T}\int_{t}^{T}El_2^{p\varepsilon}(t,s)dsdt+\varepsilon)^2-(\int_{0}^{T}\int_{t}^{T}El_2^\ast(t,s)dsdt-\int_{0}^{T}\int_{t}^{T}El_2^\varepsilon(t,s)dsdt+\varepsilon)^2\\
&=&2\langle \int_{0}^{T}\int_{t}^{T}El_2^{\varepsilon}(t,s)dsdt-\int_{0}^{T}\int_{t}^{T}El_2^{p\varepsilon}(t,s)dsdt,\int_{0}^{T}\int_{t}^{T}El_2^\ast(t,s)dsdt-\int_{0}^{T}\int_{t}^{T}El_2^\varepsilon(t,s)dsdt+\varepsilon\rangle+o(p)\\
&=&-2p[\int_{0}^{T}\int_{t}^{T}El_2^\ast(t,s)dsdt-\int_{0}^{T}\int_{t}^{T}El_2^\varepsilon(t,s)dsdt+\varepsilon]\cdot\\
&&[\int_{0}^{T}\int_{t}^{T}E(\langle l_{2x}^\varepsilon(t,s),\delta X^\varepsilon(s)\rangle+\langle l_{2y}^\varepsilon(t,s),\delta Y^\varepsilon(s)\rangle+\langle l_{2z}^\varepsilon(t,s),\delta Z^\varepsilon(s,t)\rangle +\langle l_{2u}^\varepsilon(t,s), \hat{u}^\varepsilon(s)\rangle)dsdt]+o(p).
\end{array}$$
 For the given
$\varepsilon$, we consider the following cases:

Case 1. There exists $r>0$ such that, for any $p\in (0,r),$
$$\begin{array}
[c]{ll}%
\int_{0}^{T}Eq(\psi^\ast(t))dt-\int_{0}^{T}Eq(\psi_p^\varepsilon(t))dt+\varepsilon>0,&\\
Eh(X^{\ast}(T))-Eh(X^{\varepsilon}_{p}(T))+\varepsilon>0,&\\
\int_0^TEk(Y^\ast(s))ds-\int_0^T Ek(Y^{\varepsilon}_{p}(s))ds+\varepsilon>0,&\\ \int_{0}^{T}\int_{0}^{t}El_1^\ast(t,s)dsdt-\int_{0}^{T}\int_{0}^{t}El_1^{p\varepsilon}(t,s)dsdt+\varepsilon>0,&\\
 \int_{0}^{T}\int_{t}^{T}El_2^\ast(t,s)dsdt-\int_{0}^{T}\int_{t}^{T}El_2^{p\varepsilon}(t,s)dsdt+\varepsilon>0.\end{array}$$ Then
\begin{equation*}\begin{array}{llll}
& &
\lim\limits_{p\rightarrow0}\frac{F_{\varepsilon}(\psi_{p}^{\varepsilon}(\cdot),u_{p}^{\varepsilon}(\cdot))-F_{\varepsilon}(\psi^{\varepsilon}(\cdot),u^{\varepsilon}(\cdot))}{p}
{}
\nonumber\\
& = & {}
\lim\limits_{p\rightarrow0}\frac{1}{F_{\varepsilon}(\psi_{p}^{\varepsilon}(\cdot),u_{p}^{\varepsilon}(\cdot))+F_{\varepsilon}(\psi^{\varepsilon}(\cdot),u^{\varepsilon}(\cdot))}\cdot
\frac{F_{\varepsilon}^{2}(\psi_{p}^{\varepsilon}(\cdot),u_{p}^{\varepsilon}(\cdot))-F_{\varepsilon}^{2}(\psi^{\varepsilon}(\cdot),u^{\varepsilon}(\cdot))}{p}
{}
\nonumber\\
& = & {}
\frac{1}{F_{\varepsilon}(\psi^{\varepsilon}(\cdot),u^{\varepsilon}(\cdot))} \{\langle \int_0^TEY^{\varepsilon}(t)dt-a,\int_0^TE\delta Y^{\varepsilon}(t)dt\rangle+
 \int_{0}^{T}\langle
EY^{\varepsilon}(t)-\rho(t),E\delta Y^{\varepsilon}(t)\rangle dt {}
\nonumber\\
&&-[\int_{0}^{T}Eq(\psi^\ast(t))dt-\int_{0}^{T}Eq(\psi^\varepsilon(t))dt+\varepsilon]\int_{0}^{T}E\langle
q_x(\psi^\varepsilon(t)),\hat{\psi}^\varepsilon(t)\rangle
dt\nonumber\\
& & -[Eh(X^\ast(T) )-Eh(X^{\varepsilon}(T) )+\varepsilon]E\langle
h_x(X^\varepsilon(T)),\delta X^\varepsilon(T)\rangle\\
&&-  [ \int_0^T Ek(Y^\ast(s) )ds- \int_0^T Ek(Y^{\varepsilon}(s))ds+\varepsilon] \int_0^T E\langle
k_y(Y^\varepsilon(s)),\delta Y^\varepsilon(s)\rangle ds\\
&&-[\int_{0}^{T}\int_{0}^{t}El_1^\ast(t,s)dsdt-\int_{0}^{T}\int_{0}^{t}El_1^\varepsilon(t,s)dsdt+\varepsilon][\int_{0}^{T}\int_{0}^{t}E(\langle l_{1x}^\varepsilon(t,s),\delta X^\varepsilon(s)\rangle +\langle l_{1u}^\varepsilon(t,s), \hat{u}^\varepsilon(s)\rangle)dsdt]\\
&&-[\int_{0}^{T}\int_{t}^{T}El_2^\ast(t,s)dsdt-\int_{0}^{T}\int_{t}^{T}El_2^\varepsilon(t,s)dsdt+\varepsilon][\int_{0}^{T}\int_{t}^{T}E(\langle l_{2x}^\varepsilon(t,s),\delta X^\varepsilon(s)\rangle+\langle l_{2y}^\varepsilon(t,s),\delta Y^\varepsilon(s)\rangle\\
&&+\langle l_{2z}^\varepsilon(t,s),\delta Z^\varepsilon(s,t)\rangle +\langle l_{2u}^\varepsilon(t,s), \hat{u}^\varepsilon(s)\rangle)dsdt]
 \}.
\end{array}\end{equation*}
Set

 $
\begin{array}
[c]{l}%
\bar{h}_{\varepsilon}^{0}=\frac{ \int_0^TEY^{\varepsilon}(t)dt-a}{F_{\varepsilon}(\psi^{\varepsilon}(\cdot),u^{\varepsilon}(\cdot))},\qquad
h_{\varepsilon}^{0}(t)=\frac{EY^{\varepsilon}(t)-\rho(t)}{F_{\varepsilon}(\psi^{\varepsilon}(\cdot),u^{\varepsilon}(\cdot)))},\\
\bar{h}^1_\varepsilon=-\frac{1}{F_{\varepsilon}(\psi^{\varepsilon}(\cdot),u^{\varepsilon}(\cdot))}[\int_{0}^{T}Eq(\psi^\ast(t))dt-\int_{0}^{T}Eq(\psi^\varepsilon(t))dt+\varepsilon]<0,
\\
\medskip
h^1_{\varepsilon}=-\frac{1}{F_{\varepsilon}(\psi^{\varepsilon}(\cdot),u^{\varepsilon}(\cdot))}
[Eh(X^{\ast}(T))-Eh(X^\varepsilon(T) )+\varepsilon]< 0,\\
\medskip
h^2_{\varepsilon}=-\frac{1}{F_{\varepsilon}(\psi^{\varepsilon}(\cdot),u^{\varepsilon}(\cdot))}
[ \int_0^T Ek(Y^\ast(s) )ds- \int_0^T Ek(Y^{\varepsilon}(s))ds+\varepsilon]<0,\\
\medskip
h^3_{\varepsilon}=-\frac{1}{F_{\varepsilon}(\psi^{\varepsilon}(\cdot),u^{\varepsilon}(\cdot))}
[\int_{0}^{T}\int_{0}^{t}El_1^\ast(t,s)dsdt-\int_{0}^{T}\int_{0}^{t}El_1^\varepsilon(t,s)dsdt+\varepsilon]<0,\\
\medskip
h^4_{\varepsilon}=-\frac{1}{F_{\varepsilon}(\psi^{\varepsilon}(\cdot),u^{\varepsilon}(\cdot))}
[\int_{0}^{T}\int_{t}^{T}El_2^\ast(t,s)dsdt-\int_{0}^{T}\int_{t}^{T}El_2^\varepsilon(t,s)dsdt+\varepsilon]<0,\\
\end{array} $
\\
Then it follows from (\ref{equ2.07}),
\begin{equation}\label{equ2.08}
\begin{array}
[c]{l}%
\int_{0}^{T}E\langle h_{\varepsilon}^{0}(t)+\bar{h}^0_{\varepsilon},\delta
Y^{\varepsilon}(t)\rangle
dt+\bar{h}^1_\varepsilon\int_{0}^{T}E\langle
q_x(\psi^\varepsilon(t)),\hat{\psi}^\varepsilon(t)\rangle
dt
+h^1_{\varepsilon} E\langle h_x(X^\varepsilon(T)),\delta
X^\varepsilon(T)\rangle\\
+ h^2_{\varepsilon}\int_0^TE\langle k_y(Y^\varepsilon(s)),\delta
Y^\varepsilon(s)\rangle ds
+h^3_{\varepsilon}[\int_{0}^{T}\int_{0}^{t}E(\langle l_{1x}^\varepsilon(t,s),\delta X^\varepsilon(s)\rangle +\langle l_{1u}^\varepsilon(t,s), \hat{u}^\varepsilon(s)\rangle)dsdt]\\
+h^4_{\varepsilon}[\int_{0}^{T}\int_{t}^{T}E(\langle l_{2x}^\varepsilon(t,s),\delta X^\varepsilon(s)\rangle+\langle l_{2y}^\varepsilon(t,s),\delta Y^\varepsilon(s)\rangle+\langle l_{2z}^\varepsilon(t,s),\delta Z^\varepsilon(s,t)\rangle +\langle l_{2u}^\varepsilon(t,s), \hat{u}^\varepsilon(s)\rangle)dsdt]\\
 \ge
-\sqrt{\sqrt{2}\varepsilon}[E\int_{0}^{T}|\hat{\psi}^{\varepsilon}(t)|^{2}dt+E\int_{0}^{T}|\hat{u}^{\varepsilon}(t)|^{2}dt]^{1/2}.
\end{array}
\end{equation}

Case 2.  There exists a positive sequence $\{p_{n}\}$, which
satisfies $p_{n}\rightarrow0$ such that
$$\begin{array}
[c]{ll}%
\int_{0}^{T}Eq(\psi^\ast(t))dt-\int_{0}^{T}Eq(\psi_p^\varepsilon(t))dt+\varepsilon\leq0,&\\
Eh(X^{\ast}(T))-Eh(X^{\varepsilon}_{p_n}(T))+\varepsilon\leq0, &\\
\int_0^TEk(Y^\ast(s))ds-\int_0^T Ek(Y^{\varepsilon}_{p_n}(s))ds+\varepsilon>0,&\\
\int_{0}^{T}\int_{0}^{t}El_1^\ast(t,s)dsdt-\int_{0}^{T}\int_{0}^{t}El_1^{p_n\varepsilon}(t,s)dsdt+\varepsilon\leq0,&\\
 \int_{0}^{T}\int_{t}^{T}El_2^\ast(t,s)dsdt-\int_{0}^{T}\int_{t}^{T}El_2^{p_n\varepsilon}(t,s)dsdt+\varepsilon\leq0.\end{array}$$ From the
definition of $F_\varepsilon$, for enough large $n$,\
$$F_{\varepsilon}(\psi_{p_{n}}^{\varepsilon}(\cdot),u_{p_{n}}^{\varepsilon}(\cdot))= \{|Y_{p_{n}}^{\varepsilon}(0)-\rho(0)|^{2}+\int_{0}^{T}|EY_{p_{n}}^{\varepsilon}(t)-\rho(t)|^{2}dt \}^{1/2}.
$$ Since $F_\varepsilon(\cdot)$ is continuous, we know
$F_{\varepsilon}(\psi^{\varepsilon}(\cdot),u^{\varepsilon}(\cdot))= \{|\int_0^TEY^{\varepsilon}(t)dt-a|^{2}+\int_{0}^{T}|EY^{\varepsilon}(t)-\rho(t)|^{2}dt \}^{1/2}.
$     \\ Now
\begin{equation*}\begin{array}{llll}
& &
\lim\limits_{{n}\rightarrow\infty}\frac{F_{\varepsilon}(\psi_{p_{n}}^{\varepsilon}(\cdot),u_{p_{n}}^{\varepsilon}(\cdot))-F_{\varepsilon}(\psi^{\varepsilon}(\cdot),u^{\varepsilon}(\cdot))}{p_{n}}
{}
\nonumber\\
& = & {}
\lim\limits_{{n}\rightarrow\infty}\frac{1}{F_{\varepsilon}(\psi_{p_{n}}^{\varepsilon}(\cdot),u_{p_{n}}^{\varepsilon}(\cdot))+F_{\varepsilon}(\psi^{\varepsilon}(\cdot),u^{\varepsilon}(\cdot))}\cdot
\frac{F_{\varepsilon}^{2}(\psi_{p_{n}}^{\varepsilon}(\cdot),u_{p_{n}}^{\varepsilon}(\cdot))-F_{\varepsilon}^{2}(\psi^{\varepsilon}(\cdot),u^{\varepsilon}(\cdot))}{p_{n}}
{}
\nonumber\\
& = & {}
\frac{1}{F_{\varepsilon}(\psi^{\varepsilon}(\cdot),u^{\varepsilon}(\cdot))}\{\langle \int_0^TEY^{\varepsilon}(t)dt-a,\int_0^TE\delta Y^{\varepsilon}(t) dt\rangle+\int_{0}^{T}\langle
EY^{\varepsilon}(t)-\rho(t),E\delta Y^{\varepsilon}(t)\rangle
dt\}.\end{array}
\end{equation*}
 Similar to Case 1, it follows from (\ref{equ2.07}),
\begin{equation*}
\int_{0}^{T}E\langle \bar{h}_\varepsilon ^0+h_{\varepsilon}^{0}(t),\delta
Y^{\varepsilon}(t)\rangle dt\ge
-\sqrt{\sqrt{2}\varepsilon}[E\int_{0}^{T}|\hat{\psi}^{\varepsilon}(t)|^{2}dt+E\int_{0}^{T}|\hat{u}^{\varepsilon}(t)|^{2}dt]^{1/2},
\end{equation*}
where $\bar{h}_\varepsilon
^0=\frac{\int_0^TEY^{\varepsilon}(t)dt-a}{F_{\varepsilon}(\psi^{\varepsilon}(\cdot),u^{\varepsilon}(\cdot))},\
h_{\varepsilon}^{0}(t)=\frac{EY^{\varepsilon}(t)-\rho(t)}{F_{\varepsilon}(\psi^{\varepsilon}(\cdot),u^{\varepsilon}(\cdot))},\
\bar{h}^1_\varepsilon=h_{\varepsilon}^{1}=h_{\varepsilon}^{2}=h_{\varepsilon}^{3}=h_{\varepsilon}^{4}=0. $

Similarly, we can prove $(\ref{equ2.08})$ still holds for the other
thirty cases.

In summary for given $\varepsilon$, we have
$$
\begin{array}
[c]{ll}%
\rm(i)\quad & (\ref{equ2.08})\ \rm holds,\\
\rm(ii)\quad & \bar{h}^1_{\varepsilon}\leq0,\ h^1_{\varepsilon}\leq0,\ h^2_{\varepsilon}\leq0,\\
\rm(iii)\quad &
|h_{\varepsilon}^{0}|^2+\int_{0}^{T}|h_{\varepsilon}^{0}(t)|^{2}dt+|\bar{h}^1_{\varepsilon}|^{2}+|h^1_{\varepsilon}|^{2}+|h^2_{\varepsilon}|^{2}+|h^3_{\varepsilon}|^{2}+|h^4_{\varepsilon}|^{2}=1.
\end{array}
$$
 Hence there is a subsequence
$(\bar{h}^0_{\varepsilon_{n}},h^0_{\varepsilon_{n}}(\cdot),\bar{h}^1_{\varepsilon_{n}},h^1_{\varepsilon_{n}},h^2_{\varepsilon_{n}},h^3_{\varepsilon_{n}},h^4_{\varepsilon_{n}})$
of
$(\bar{h}^0_{\varepsilon},h^0_{\varepsilon}(\cdot),\bar{h}^1_{\varepsilon},h^1_{\varepsilon},h^2_{\varepsilon},h^3_{\varepsilon},h^4_{\varepsilon})$,
such that $\bar{h}^0_{\varepsilon_{n}}\rightarrow  \bar{h}_0,\
h^0_{\varepsilon_{n}}(\cdot)$ $\rightarrow h_0(\cdot),\
\bar{h}^1_{\varepsilon_{n}}\rightarrow \bar{h}_1,\
h^1_{\varepsilon_{n}}\rightarrow h_1,\
h^2_{\varepsilon_{n}}\rightarrow h_2,\
h^3_{\varepsilon_{n}}\rightarrow h_3,\
h^4_{\varepsilon_{n}}\rightarrow h_4$. Since $
\bar{h}^1_{\varepsilon},\ h^1_{\varepsilon},\
h^2_{\varepsilon},\
h^3_{\varepsilon},\
h^4_{\varepsilon}\leq0$, we have $\bar{h}_1,\ h_1,\ h_2,\ h_3,\ h_4\leq0$.

Because of
$d((\psi^{\varepsilon}(\cdot),u^{\varepsilon}(\cdot)),(\psi^{\ast}(\cdot),u^{\ast}(\cdot)))\leq\sqrt{2}\varepsilon$,
we have
$(\psi^{\varepsilon}(\cdot),u^{\varepsilon}(\cdot))\rightarrow(\psi^{\ast}(\cdot),
u^{\ast}(\cdot))$ in $\mathcal {U}$. Therefore, from the
wellposedness of FBSVIEs, it is easy to check $\delta
X^\varepsilon(\cdot)\rightarrow \delta X(\cdot),\ \delta
Y^\varepsilon(\cdot)\rightarrow \delta Y(\cdot)$, as
$\varepsilon\rightarrow0.$ Furthermore,  as
$\varepsilon\rightarrow0$
\begin{eqnarray*}
& &| E\langle h_{x}(X^\varepsilon(T)),\delta
X^\varepsilon(T)\rangle - E \langle h_{x}(X^\ast(T)),\delta
X(T)\rangle | {}
\nonumber\\
& = & {} | E\langle h_{x}(X^\varepsilon(T)),\delta
X^\varepsilon(T)-\delta X(T)\rangle + E \langle
h_{x}(X^\varepsilon(T))-h_{x}(X^\ast(T)),\delta X(T)\rangle
|\rightarrow0.
\end{eqnarray*}
Indeed, together with the Schwarz inequality, using the boundedness
of $h_{x}$, we can get the limit of the first part goes to 0; from
the continuity $h_x$, we get the second part also goes to 0.
Similarly, as
$\varepsilon\rightarrow 0$, we have $\int_{0}^{T}E\langle
q_x(\psi^\varepsilon(t)),\hat{\psi}^\varepsilon(t)\rangle
dt\rightarrow \int_{0}^{T}E\langle
q_x(\psi^\ast(t)),\hat{\psi}(t)\rangle dt,$   $ \langle
k_y(Y^\varepsilon(0)),\delta Y^\varepsilon(0)\rangle\rightarrow
\langle k_y(Y^\ast(0)),\delta Y(0)\rangle$ and \\
$\begin{array}{l}\int_{0}^{T}\int_{0}^{t}E\langle l_{1x}^\varepsilon(t,s),\delta X^\varepsilon(s)\rangle dsdt\rightarrow \int_{0}^{T}\int_{0}^{t}E\langle l_{1x}^\ast(t,s),\delta X(s)\rangle dsdt,\\
\int_{0}^{T}\int_{0}^{t}E\langle l_{1u}^\varepsilon(t,s),\hat{u}^\varepsilon(s) \rangle dsdt \rightarrow \int_{0}^{T}\int_{0}^{t}E\langle l_{1u}^\ast(t,s),\hat{u}^\ast(s)\rangle dsdt,\\
\int_{0}^{T}\int_{t}^{T}E\langle l_{2x}^\varepsilon(t,s),\delta X^\varepsilon(s)\rangle dsdt\rightarrow \int_{0}^{T}\int_{t}^{T}E\langle l_{2x}^\ast(t,s),\delta X(s)\rangle dsdt,\\
\int_{0}^{T}\int_{t}^{T}E\langle l_{2y}^\varepsilon(t,s),\delta Y^\varepsilon(s)\rangle dsdt\rightarrow \int_{0}^{T}\int_{t}^{T}E\langle l_{2y}^\ast(t,s),\delta Y(s)\rangle dsdt,\\
\int_{0}^{T}\int_{t}^{T}E\langle l_{2z}^\varepsilon(t,s),\delta Z^\varepsilon(s,t)\rangle dsdt\rightarrow \int_{0}^{T}\int_{t}^{T}E\langle l_{2z}^\ast(t,s),\delta Z(s,t)\rangle dsdt,\\
\int_{0}^{T}\int_{t}^{T}E\langle l_{2u}^\varepsilon(t,s),\hat{u}^\varepsilon(s) \rangle dsdt\rightarrow \int_{0}^{T}\int_{t}^{T}E\langle l_{2u}^\ast(t,s),\hat{u}^\ast(s)\rangle dsdt.\end{array}$

Let $\varepsilon\rightarrow0$ in (\ref{equ2.08}), the result holds.
The proof is completed.
\end{proof}

\subsection{Maximal principle}

We introduce the adjoint equation:
\begin{equation}  \label{equ2.09}
\left\{%
\begin{array}{llll}
m(t) =  A(t)+{ \int_{t}^{T}}\left[b_x^\ast(s,t)^Tm(s)+%
\sigma_x^\ast(s,t)^Tn(s,t)\right]ds-{ \int_{t}^{T}}n(t,s)dB_s, &
\\
p(t)  =  B(t)+ \int_{0}^{t}g^{\ast}_{y}(s,t)^{T}p(s)ds +{%
 \int_{0}^{t}}E[g^{\ast}_{\zeta}(s,t)^{T}|\mathcal{F}%
_s]p(s)dB_s,\quad t\in[0,T], &
\end{array}
\right.
\end{equation}
where
\begin{equation*}
\left\{%
\begin{array}{llll}
A(t)  =  h_1b^\ast_x(T,t)^Th_x(X^\ast(T))+h_1
\sigma_{x}^{\ast}(T,t)^{T}\pi(t)   +\int_{0}^{t}g_x^\ast(s,t)^Tp(s)ds, &  \\
B(t)  = h_0(t)+ \bar{h}_0+h_2k_y(Y^\ast(t)),
\end{array}%
\right.
\end{equation*}
and $h_x(X^\ast(T))=Eh_x(X^\ast(T))+\int_{0}^{T}\pi(s)dB_s.$

%\begin{remark}
%By the definition of M-solution for BSVIE, we obtain
%\begin{equation*}
%\delta Y(t)=E\delta Y(t)+\int_{0}^{t}\delta Z(t,s)dB_{s},\
%m(t)=Em(t)+\int_{0}^{t}n(t,s)dB_{s},\ t\in \lbrack 0,T].
%\end{equation*}%
%Thus, we have $\delta Z_{i}(t,s)=E^{\mathcal{F}_{s}}D_{s}^{i}\delta Y(t),\
%n_{i}(t,s)=E^{\mathcal{F}_{s}}D_{s}^{i}m(t),\ i=1,\cdots ,d.$ Specially, $%
%\delta Z_{i}(s,0)=(E[D_0\delta Y(s)],\cdots ,E[D_0\delta Y(s)]),$ $n(s,0)=(E[D_0m(s)],%
%\cdots ,E[D_0m(s)]).$ Note that $D$ is a malliavin operator. For the details the
%readers can refer to \cite{N}.
%\end{remark}

By the duality principles, we get the following theorem:

\begin{theorem}
\label{TH1} Assume that $\left( \mathbf{A}_{1}\right) -\left( \mathbf{A}%
_{4}\right) $ hold and $l_{1},\ l_{2}=0$. Let $(\psi ^{\ast
}(\cdot),u^{\ast }(\cdot))$ be the optimal control pair;
$(X^{\ast }(\cdot ),Y^{\ast }(\cdot ),Z^{\ast }(\cdot ,\cdot ))$ be
the corresponding optimal trajectory. Then there exist a
deterministic function $h_{0}(\cdot )\in \mathbb{R}^{m}$,
$\bar{h}_{0}\in
\mathbb{R}^{m},\ \bar{h}_{1},\ h_{1},\ h_{2}\leq 0$ such that $\forall (\psi(\cdot),u(\cdot))\in \mathcal{U},$
\begin{equation*}
\begin{array}{llll}
&  & \langle p(t)+\bar{h}_{1}q_{x}(\psi ^{\ast }(t)),\psi(t)-\psi
^{\ast }(t)\rangle +{ \int_{t}^{T}}\langle g_{u}^{\ast
}(t,s)^{T}p(t),u(s)-u^{\ast }(s)\rangle ds &  \\
&  & +h_{1}\langle b_{u}^{\ast }(T,t)^{T}h_{x}(X^{\ast
}(T))+\sigma _{u}^{\ast }(T,t)^{T}\pi
(t),u(t)-u^{\ast }(t)\rangle &  \\
&  & +{ \int_{t}^{T}}\langle b_{u}^{\ast
}(s,t)^{T}m(s)+\sigma _{u}^{\ast
}(s,t)^{T}n(s,t),u(t)-u^{\ast }(t)\rangle ds\geq 0,\
a.e., a.s.
\end{array}%
\end{equation*}%
where $(m(\cdot ),n(\cdot ,\cdot ),p(\cdot ))$ is the solution of
the adjoint equation~(\ref{equ2.09}).
\end{theorem}

\begin{proof}

From the duality principles (Lemma \ref{th11}, Lemma \ref{th12}), we
have the following relations

$\begin{array}{llll}&&E\int_{0}^{T}\langle
A(t),\delta X(t)\rangle dt\\ &=&E\int_{0}^{T}\langle
m(t),\int_{0}^{T}b_u^\ast(t,s)\hat{u}(s)ds+\int_{0}^{T}\sigma_u^\ast(t,s)\hat{u}(s)dB_s\rangle
dt\\
&=&E\int_{0}^{T}{ \int_{t}^{T}}m(s)^Tb_u^\ast(s,t)\hat{u}(t)
dsdt+E\int_{0}^{T}{ \int_{t}^{T}}\sum\limits_{i=1}^dn_i(s,t)^T(\sigma_u^\ast)^i(s,t)\hat{u}(t)
dsdt  \end{array}$\\
and $$\begin{array}{llll}&&E\int_{0}^{T}\langle
B(t),\delta Y(t)\rangle dt =E\int_{0}^{T}\langle
p(t),\hat{\psi}(t)+{ \int_{t}^{T}}[g_x^\ast(t,s)\delta
X(s)+g_u^\ast(t,s)\hat{u}(s)]ds\rangle dt.\end{array}$$ Combined
with the variational inequality (Theorem \ref{th3}), we get
\\
$
\begin{array}{llll}
 0 &\leq&  \int_{0}^{T}E\langle \bar{h}_{0}+h_{0}(t),\delta
Y(t)\rangle dt+\bar{h}_1\int_{0}^{T}E\langle
q_x(\psi^\ast(t)),\hat{\psi}(t)\rangle dt+h_1 E\langle
h_{x}(X^{\ast}(T)),\delta X(T)\rangle
\\
&&+h_2 \int_0^TE\langle k_{y}(Y^{\ast }(s)),\delta Y(s)\rangle ds  \\
&=& \int_{0}^{T}E \langle \bar{h}_{0}+h_{0}(t),\delta Y(t)\rangle
dt+\bar{h}_1\int_{0}^{T}E\langle
q_x(\psi^\ast(t)),\hat{\psi}(t)\rangle dt+h_1 E\langle
h_{x}(X^{\ast}(T)),\delta X(T)\rangle  \\
&&+h_2 \int_0^TE\langle k_{y}(Y^{\ast }(s)),\delta Y(s)\rangle ds +E\int_{0}^{T}{ \int_{t}^{T}}m(s)^Tb_u^\ast(s,t)\hat{u}(t)
dsdt\\
&&+E\int_{0}^{T}{ \int_{t}^{T}}\sum\limits_{i=1}^dn_i(s,t)^T(\sigma_u^\ast)^i(s,t)\hat{u}(t)
dsdt-E\int_{0}^{T}\langle A(t),\delta X(t)\rangle
dt+E\int_{0}^{T}\langle
p(t),\hat{\psi}(t)\\
&&+{ \int_{t}^{T}}[g_x^\ast(t,s)\delta
X(s)+g_u^\ast(t,s)\hat{u}(s)]ds\rangle dt-E\int_{0}^{T}\langle B(t),\delta Y(t)\rangle dt\\
&=&  \int_{0}^{T}E\langle \bar{h}_{0}+h_{0}(t),\delta Y(t)\rangle
dt+\bar{h}_1\int_{0}^{T}E\langle
q_x(\psi^\ast(t)),\hat{\psi}(t)\rangle dt\\
&&+h_1E \langle
h_{x}(X^{\ast}(T)),\int_{0}^{T}b_u^\ast(T,s)\hat{u}(s)ds+\int_{0}^{T}\sigma_u^\ast(T,s)\hat{u}(s)dB_s\rangle\\
&&+h_1E \langle
h_{x}(X^{\ast}(T)),\int_{0}^{T}b_x^\ast(T,s)\delta X(s)ds+\int_{0}^{T}\sigma_x^\ast(T,s)\delta X(s)dB_s\rangle\\
&&+h_2 \int_0^TE\langle k_{y}(Y^{\ast }(s)),\delta Y(s)\rangle ds \\
&&+E\int_{0}^{T}{ \int_{t}^{T}}m(s)^Tb_u^\ast(s,t)\hat{u}(t)
dsdt+E\int_{0}^{T}{ \int_{t}^{T}}\sum\limits_{i=1}^dn_i(s,t)^T(\sigma_u^\ast)^i(s,t)\hat{u}(t)
dsdt\\
&&-E\int_{0}^{T}\langle A(t),\delta X(t)\rangle
dt+E\int_{0}^{T}\langle
p(t),\hat{\psi}(t)+{ \int_{t}^{T}}[g_x^\ast(t,s)\delta
X(s)+g_u^\ast(t,s)\hat{u}(s)]ds\rangle dt\\
&&-E\int_{0}^{T}\langle B(t),\delta Y(t)\rangle dt \\
&=&E \int_{0}^{T}\langle p(t)+\bar{h}_1
q_x(\psi^\ast(t)),\hat{\psi}(t)\rangle dt+E
\int_{0}^{T}
{ \int_{t}^{T}}\langle g_u^\ast(t,s)^Tp(t), \hat{u}(s)\rangle dsdt\\
&&+h_1E\int_{0}^{T} \langle
 b_u^\ast(T,t)^Th_{x}(X^{\ast}(T))+\sum\limits_{i=1}^d(\sigma_u^\ast)^i(T,t)^T\pi_i(t),\hat{u}(t)\rangle dt\\
&&+E\int_{0}^{T}{ \int_{t}^{T}}\langle
b_u^\ast(s,t)^Tm(s)+\sum\limits_{i=1}^d(\sigma_u^\ast)^i(s,t)^Tn_i(s,t),\hat{u}(t)\rangle
dsdt.
\end{array}
$\\
 Since the above holds for all $(\psi(\cdot),u(\cdot))\in \mathcal{U}$, we obtain
$$
\begin{array}{llll}&&\langle
p(t)+\bar{h}_1 q_x(\psi^\ast(t)),\psi(t)-\psi^\ast(t)\rangle +
{ \int_{t}^{T}}\langle g_u^\ast(t,s)^Tp(t), u(s)-u^\ast(s)\rangle ds\\
&&+h_1 \langle
 b_u^\ast(T,t)^Th_{x}(X^{\ast}(T))+\sigma_u^\ast(T,t)^T\pi(t),u(t)-u^\ast(t)\rangle \\
&&+{ \int_{t}^{T}}\langle
b_u^\ast(s,t)^Tm(s)+\sigma_u^\ast(s,t)^Tn(s,t),u(t)-u^\ast(t)\rangle
ds\geq0,\ a.e., a.s.\end{array} $$
\end{proof}

When $l_1,\ l_2\neq 0,$ the associated adjoint equation is:
\begin{equation}  \label{equ2.10}
\left\{%
\begin{array}{llll}
m(t) = A(t)+{ \int_{t}^{T}}[b_x^\ast(s,t)^Tm(s)+\sigma_x^\ast(s,t)^Tn(s,t)]ds-{ \int_{t}^{T}%
}n(t,s)dB_s, &  \\
p(t)  =  B(t)+ \int_{0}^{T}g^{\ast}_{y}(s,t)^{T}p(s)ds +{%
 \int_{0}^{t}}E[g^{\ast}_{\zeta}(s,t)^{T}|\mathcal{F}%
_s]p(s)dB_s,\quad t\in[0,T], &
\end{array}%
\right.
\end{equation}
where
\begin{equation*}
\left\{%
\begin{array}{llll}
A(t)  = h_1\sigma^{\ast}_{x}(T,t)^{T}%
\pi(t)+h_1b^\ast_x(T,t)^Th_x(X^\ast(T)) +\int_{0}^{T}g_x^\ast(s,t)^Tp(s)ds&  \\
 \qquad \ \ \ +{
\int_{t}^{T}}h_3l_{1x}^\ast(s,t)ds+{ \int_{t}^{T}}%
h_4l_{2x}^\ast(s,t)ds, &  \\
B(t)  =  h_0(t)+\bar{h}_0+h_2k_y(Y^\ast(t))+\int_{0}^{T}h_4l_{2y}^\ast(s,t)ds+{
\int_{0}^{t}}h_4l_{2\zeta}^\ast(s,t)ds, &
\end{array}%
\right.
\end{equation*}
and $h_x(X^\ast(T))=Eh_x(X^\ast(T))+\int_{0}^{T}\pi(s)dB_s.$

Similarly, we have the following maximum principle:

\begin{theorem}
Assume $\left( \mathbf{A}_{1}\right) -\left( \mathbf{A}_{4}\right) $
hold. Let $(\psi ^{\ast }(\cdot),u^{\ast }(\cdot))$ be the
optimal control pair; $(X^{\ast }(\cdot ),Y^{\ast }(\cdot ),$
$Z^{\ast }(\cdot ,\cdot ))$ be the corresponding optimal trajectory.
Then there exist a deterministic function $h_{0}(\cdot )\in
\mathbb{R}^{m}$, $\bar{h}_{0}\in \mathbb{R}^{m},\ \bar{h}_{1},\
h_{1},\ h_{2},\ h_{3},$ $h_{4}\leq 0$ such that $\forall
(\psi(\cdot),u(\cdot))\in U,$
\begin{equation*}
\begin{array}{llll}
&  & \langle p(t)+\bar{h}_{1}q_{x}(\psi ^{\ast }(t)),\psi(t)-\psi
^{\ast }(t)\rangle +{ \int_{t}^{T}}\langle g_{u}^{\ast
}(t,s)^{T}p(t),u(s)-u^{\ast }(s)\rangle ds &  \\
&  & +h_{1}\langle b_{u}^{\ast }(T,t)^{T}h_{x}(X^{\ast
}(T))+\sigma _{u}^{\ast }(T,t)^{T}\pi
(t),u(t)-u^{\ast }(t)\rangle &  \\
&  & +{ \int_{t}^{T}}\langle b_{u}^{\ast
}(s,t)^{T}m(s)+\sigma _{u}^{\ast
}(s,t)^{T}n(s,t),u(t)-u^{\ast }(t)\rangle ds &  \\
&  & +h_{3}\int_{0}^{T}\langle l_{1u}^{\ast
}(t,s),u(s)-u^{\ast }(s)\rangle ds+h_{4}{ \int_{t}^{T}}%
\langle l_{2u}^{\ast }(t,s),u(s)-u^{\ast }(s)\rangle ds\geq 0,\
a.e., a.s. &
\end{array}%
\end{equation*}%
where $(m(\cdot ),n(\cdot ,\cdot ),p(\cdot ))$ is the solution of the
adjoint equation~(\ref{equ2.10}).
\end{theorem}

\begin{remark}
When the terminal condition $\psi (\cdot)$ is replaced by $\psi
(\cdot )+\varphi (X(T))$ in (\ref{equ1.1}), the above methods can still
go through.

 %Our problem contains the optimal
%control problem in Yong \cite{JY} as a special case, i.e., $%
%\psi \equiv 0,\ g\equiv 0,\ l_{2}\equiv 0,\ q\equiv 0,\ k\equiv 0$,
%and considering the classical process $u(\cdot )$. It is easy to
%check that, in this case, the
%above maximum principle is consistent with the one obtained in Yong \cite{JY}%
%.
\end{remark}

\section{Examples}

First we will give an example  associated with the model studied above.

\begin{example}\rm
Consider the following controlled system ($m=d=1$):
\begin{equation}
\left\{
\begin{array}{llll}
X(t)  =  \int_{0}^{T}tu(s)dB_{s},\ \  &  \\
Y(t)  =  \psi (t)+{ \int_{t}^{1}}(t-1)u(s)ds-{
\int_{t}^{1}}Z(t,s)dB_{s},\ t\in \lbrack 0,1], &
\end{array}%
\right.   \label{equ5.5}
\end{equation}%
with the control domain
\begin{equation*}
\mathcal{U}=\{ (\psi(\cdot) ,u(\cdot))|\psi (\cdot )\in L_{\mathcal{F}%
_{T}}^{2}(0,1),\  u(\cdot )\in L_{\mathbb{F}%
}^{2}(0,1),\ \psi (t )\in \lbrack 0,1],\ u(t )\in \lbrack
-\frac{1}{2},1],\ a.e., a.s.\}
\end{equation*}%
and the objective function
\begin{equation}
J(\psi(\cdot) ,u(\cdot))=E\left\{ X(1)^{2}+Y(0)\right\} .
\label{equ6.5}
\end{equation}%
We will minimize the objective function under the constraints
$(\psi(\cdot) ,u(\cdot))\in \mathcal{U}.$ After substituting $X(1),\
Y(0)$ into the objective function, we get
\begin{equation}
J(\psi(\cdot) ,u(\cdot))=E[ { \int_{0}^{1}}u(s)^{2}ds+\psi (0)-{%
 \int_{0}^{1}}u(s)ds] .  \label{equ7.5}
\end{equation}%
From (\ref{equ7.5}), we obtain the optimal control:
\begin{equation*}
\psi ^{\ast }(s)=%
\begin{cases}
0, & s=0, \\
\mbox{values in }\lbrack 0,1], & s\in (0,1],%
\end{cases}%
\quad u^{\ast }(s)=\frac{1}{2},\ s\in \lbrack 0,1].
\end{equation*}%
So, $\min\limits_{(\psi(\cdot) ,u(\cdot))\in \mathcal
{U}}J(\psi(\cdot) ,u(\cdot))=-\frac{1}{4}.$
\end{example}

At last  we give an example to show the form of the optimal terminal $\psi (\cdot )$.

\begin{example}\rm
For convenience, we suppose $m=d=1,$ and consider a simple BSVIE as follows:
\begin{equation*}
\begin{array}{llll}
Y(t)=\psi (t)+{ \int_{t}^{1}}[AY(s)+BZ(s,t)]ds-{
\int_{t}^{1}}Z(t,s)dB_{s},\quad t\in \lbrack 0,1], &  &  &
\end{array}%
\end{equation*}%
$A,\ B\in \mathbb{R}.$ We will maximize the objective function $J(\psi(\cdot) )=%
\frac{1}{2}E\left[ { \int_{0}^{1}}\psi (s)^{2}ds\right]
$, subject to $\psi (\cdot )\in L_{\mathcal{F}_{T}}^{2}(0,1),\ \psi
(t )\in \lbrack 0,1],\ EY^{\psi }(t)=\rho
(t),\ t\in \lbrack 0,1],\ a.e., a.s.$

From Subsection 3.4, we know the adjoint process $p(\cdot)$
satisfies
\begin{equation*}
\begin{array}{llll}
p(t) & = & h_0(t)+(A+B)\bar{h}_0 +\int_{0}^{T}Ap(s)ds+{%
 \int_{0}^{t}}Bp(s)dB_s,\quad t\in[0,1]. &
\end{array}%
\end{equation*}
Applying Theorem \ref{TH1}, we have, if $\psi^\ast(\cdot)$ is
optimal to $J(\psi(\cdot))$, then there exists  a
deterministic function $h_0(\cdot)$, and $\bar{h}_1\leq0,\
|h_0(\cdot)|+|\bar{h}_1|\neq0$ such that, for any $\psi(\cdot),$
\begin{equation*}
(p(t)+\bar{h}_1\psi^\ast(t))(\psi(t)-\psi^\ast(t))\geq0,\
t\in[0,1],\ a.s.
\end{equation*}
Similar to the example in Ji, Zhou \cite{SJXZ}, let $\Omega_1:=
\{(\omega,t)\in\Omega \times[0,1]| \psi^\ast(t,\omega)=0\}$,
$\Omega_2:=
\{(\omega,t)\in\Omega \times[0,1]| \psi^\ast(t,\omega)=1\}$, we obtain $%
\psi^\ast(\cdot)$ satisfies
\begin{equation*}
\begin{array}{llll}
&  & p(t)+\bar{h}_1\psi^\ast(t)\geq0,\ (\omega,t)\in \Omega_1,\ a.s. &  \\
&  & p(t)+\bar{h}_1\psi^\ast(t)\leq0,\ (\omega,t)\in \Omega_2,\ a.s. &  \\
&  & p(t)+\bar{h}_1\psi^\ast(t)= 0,\ (\omega,t)\in
\Omega-\Omega_1-\Omega_2,\ a.s. &
\end{array}%
\end{equation*}
\end{example}

\section*{Acknowledgements}
The authors thank Prof. Jiongmin Yong and Dr. Tianxiao Wang for their
helpful suggestions and discussions.

This work was supported by National Natural Science Foundation of
China (No. 11171187, No. 10871118 and No. 10921101); supported by
the Program of Introducing Talents of Discipline to Universities of
China (No.B12023); supported by Program for New Century Excellent
Talents in University of China.

\end{document}